\documentclass[article]{imsart}
\RequirePackage[OT1]{fontenc}
\RequirePackage{amsthm,amsmath,amssymb,mathrsfs,mathdots,mathdots,mathtools}
\RequirePackage{natbib}
\RequirePackage[colorlinks,citecolor=blue,urlcolor=blue]{hyperref}
\usepackage{txfonts}
\usepackage{graphicx}
\usepackage{enumitem}
\usepackage{booktabs}
\usepackage{float}
\usepackage{color,xcolor}

\startlocaldefs
\numberwithin{equation}{section}
\theoremstyle{plain}
\newtheorem{theo}{Theorem}[section]
\newtheorem{lem}[theo]{Lemma}
\newtheorem{sat}[theo]{Proposition}
\newtheorem{defi}[theo]{Definition}
\newtheorem{korr}[theo]{Corollary}
\theoremstyle{remark}
\newtheorem{rem}{Remark}[section]
\newtheorem{cond}{Condition}[section]
\endlocaldefs

\def\thick#1{\hbox{\rlap{$#1$}\kern0.25pt\rlap{$#1$}\kern0.25pt$#1$}}

\usepackage{dsfont}
\DeclarePairedDelimiter\ceil{\lceil}{\rceil}
\DeclarePairedDelimiter\floor{\lfloor}{\rfloor}

\newlist{inparaenum}{enumerate}{2}
\setlist[inparaenum]{nosep}
\setlist[inparaenum,1]{label=(\roman*)}
\setlist[inparaenum,2]{label=(\roman{inparaenumi}\emph{\alph*})}

\def\thick#1{\hbox{\rlap{$#1$}\kern0.25pt\rlap{$#1$}\kern0.25pt$#1$}}

\newcommand{\ba}{{\boldsymbol{a}}}
\newcommand{\bb}{{\boldsymbol{b}}}
\newcommand{\bc}{{\boldsymbol{c}}}
\newcommand{\bd}{{\boldsymbol{d}}}
\newcommand{\be}{{\boldsymbol{e}}}

\newcommand{\bu}{{\boldsymbol{u}}}
\newcommand{\bv}{{\boldsymbol{v}}}
\newcommand{\bw}{{\boldsymbol{w}}}
\newcommand{\bx}{{\boldsymbol{x}}}
\newcommand{\by}{{\boldsymbol{y}}}
\newcommand{\bt}{{\boldsymbol{t}}}
\newcommand{\bz}{{\boldsymbol{z}}}

\newcommand{\bM}{{\boldsymbol{M}}}
\newcommand{\bN}{{\boldsymbol{N}}}

\newcommand{\bR}{{\boldsymbol{R}}}
\newcommand{\bS}{{\boldsymbol{S}}}
\newcommand{\bX}{{\boldsymbol{X}}}

\newcommand{\bU}{{\boldsymbol{U}}}

\newcommand{\bY}{{\boldsymbol{Y}}}
\newcommand{\bZ}{{\boldsymbol{Z}}}

\newcommand{\bbeta}{{\boldsymbol{\eta}}}

\def\bzero{{\bf 0}}
\def\binf{{\boldsymbol \infty}}
\def\bone{{\bf 1}}
\def\indic{\mathds{1}}

\def\real{{\mathbb R}}
\def\reald{{\mathbb R}^d}

\def\expect{{\textrm E}}
\def\prob{{\mathbb{P}}}

\newcommand{\pk}[1]{\mathbb{P} \left\{#1 \right\} }

\def\setU{\mathcal{{U}}}
\def\setV{\mathcal{{V}}}

\newcommand{\ldot}{,\ldots,}

\newcommand{\limit}[1]{\lim_{#1 \to   \infty}}

\def\FRE{\mbox{Fr\'{e}chet} }

\def\slowf{\mathscr{L}}

\def\simp{\mathcal{S}_{d}}
\def\preal{\mathcal{Y}}
\def\sreal{\mathcal{X}}

\newcommand{\R}{\mathbb{R}}

\newcommand{\N}{\mathbb{N}}
\newcommand{\inr}{\in \R}

\def\IF{\infty}

\def\conAS{\stackrel{\text{as}}{\longrightarrow}}
\def\conP{\stackrel{\text{p}}{\longrightarrow}}

\def\GPWM{{\scriptscriptstyle \mathrm{GPWM}}}
\def\ML{{\scriptscriptstyle \mathrm{ML}}}
\def\MD{{\scriptscriptstyle\mathrm{MD}}}
\def\PICK{{\scriptscriptstyle\mathrm{P}}}
\def\CFG{{\scriptscriptstyle\mathrm{CFG}}}
\def\picka{A_{\alpha}}
\def\pickainv{A^\star}
\def\estpicka{\widehat{A}_{\alpha,n}}
\def\estpickin{\widehat{A^{\star}}_n}
\def\estlambda{\widehat{\vartheta}}
\def\estmado{\widehat{\nu}_{n}}
\def\estalpha{{\widehat{\alpha}_n}}

\DeclareMathOperator*{\argmax}{\arg\!\max}

\newcommand{\norm}[1]{\lVert{#1}\rVert}

\newcommand{\Cov}{\operatorname{Cov}}
\newcommand{\diff}{\mathrm{d}}

\newcommand{\kb}[1]{\boldsymbol{#1}}
\newcommand{\vk}[1]{\kb{#1}}

\newcommand{\EE}[1]{\textrm{E}\left( #1 \right)}

\def\IF{\infty}

\def\MDA{\mathcal{D}}
\def\LAP{\mathsf{L}}

\newcommand{\BQN}{\begin{eqnarray}}
\newcommand{\EQN}{\end{eqnarray}}
\newcommand{\BQNY}{\begin{eqnarray*}}
\newcommand{\EQNY}{\end{eqnarray*}}
\def\nncol#1{\textcolor{c30}{#1}}
\def\nncol#1{#1}

\def\bqny#1{ \nncol{ \begin{eqnarray*} #1 \end{eqnarray*}}}

\newcommand{\BS}{\begin{sat}}
\newcommand{\ES}{\end{sat}}
\newcommand{\BT}{\begin{theo}}
\newcommand{\ET}{\end{theo}}
\newcommand{\BK}{\begin{korr}}
\newcommand{\EK}{\end{korr}}

\newcommand{\BD}{\begin{de}}
\newcommand{\ED}{\end{de}}

\newcommand{\BIT}{\begin{itemize}}
\newcommand{\EIT}{\end{itemize}}
\newcommand{\BDI}{\begin{description}}
\newcommand{\EDI}{\end{description}}

\newcommand{\BRM}{\begin{remarks}}
\newcommand{\ERM}{\end{remarks}}

\newcommand{\BEL}{\begin{lem}}
\newcommand{\EEL}{\end{lem}}

\def\cbbridge{\mathbb{B}}

\def\cproc{\mathbb{C}}

\startlocaldefs
\numberwithin{equation}{section}
\theoremstyle{plain}

\endlocaldefs

\begin{document}

\begin{frontmatter}
\title{Multivariate Extremes Over a Random Number of Observations}
\runtitle{Extremes on a Random Number of Observations}

\begin{aug}
\author{\fnms{Enkelejd} \snm{Hashorva}\thanksref{a,e1}\ead[label=e1,mark]{enkelejd.hashorva@unil.ch}}
\author{\fnms{Simone A.} \snm{Padoan}\thanksref{b,e2}\ead[label=e2,mark]{simone.padoan@unibocconi.it}}
\and
\author{\fnms{Stefano} \snm{Rizzelli}\thanksref{c, e3}%
\ead[label=e3,mark]{stefano.rizzelli@epfl.ch}%
}

\address[a]{Department of Actuarial Science, University of Lausanne, UNIL-Dorigny 1015 Lausanne, Switzerland.\\
\printead{e1}}

\address[b]{Department of Decision Sciences,
Bocconi University, via Roentgen 1, 20136 Milan, Italy.\\
\printead{e2}}

\address[c]{Institute of Mathematics, \'Ecole Polytechnique F\'ed\'erale de Lausanne, CH-1015 Lausanne, Switzerland.\\
	\printead{e3}}

\runauthor{Hashorva et al.}

\end{aug}

\begin{abstract}
The classical multivariate extreme-value theory concerns the modeling of extremes in a multivariate random sample, suggesting the use of max-stable distributions. 
In this work, the classical theory is extended to the case where aggregated data, such as maxima of a random number of observations, are considered.  
We derive a limit theorem concerning the 
attractors
for the distributions of the aggregated data, which boil down to a new family of max-stable distributions. 
We also connect the extremal dependence structure of classical max-stable distributions and that of our new family of max-stable distributions.
By means of an inversion method, we derive a semiparametric composite-estimator for the extremal dependence of the unobservable data, starting from a preliminary estimator of the extremal dependence of the aggregated data. Furthermore, we develop the large-sample theory of the composite-estimator and illustrate its finite-sample performance via a simulation study.
\end{abstract}

\begin{keyword}[class=MSC]
\kwd[Primary ]{62G32}
\kwd{62G05}
\kwd{62G20}
\kwd[; secondary ]{60F05}
\kwd{60G70}
\end{keyword}

\begin{keyword}
\kwd{Extremal dependence}
\kwd{Extreme-value copula}
\kwd{Inverse problem}
\kwd{Multivariate max-stable distribution}
\kwd{Nonparametric estimation}
\kwd{Pickands dependence function}
\end{keyword}
\end{frontmatter}

\section{Introduction and background}\label{sec:intro}
   
The multivariate extreme-value theory aims to quantify the probability of extreme events concerning multiple dependent observations.
A commonly employed approach for modelling extremes in high dimensions is the componentwise maxima, where, for each of the involved variables, the partial maximum values are taken into account, e.g., yearly maxima, \citep[e.g., ][Ch. 4]{r6}.
Basic foundations of the componentwise maxima approach are here briefly introduced. 

First, however, we specify the notation
that we use throughout  the paper.
Given $\sreal\subset \real^n$,  $n\in \N$, let $\ell^{\infty}(\sreal)$ denote the spaces of 
bounded real-valued functions on $\sreal$.
For $f: \sreal\mapsto  \real$, let $\|f\|_{\infty}=\sup_{\bx \in \sreal} |f(\bx)|$. 
The arrows ``$\conAS$'', `$\conP$'', ``$\rightsquigarrow$'' denote convergence (outer) almost surely, convergence in (outer) probability and convergence in distribution of random vectors \citep[see][Ch. 2]{r2} or weak convergence of random functions in $\ell^{\infty}(\sreal)$ \citep[see][Ch. 18--19]{r2},
the distinction between the two will be clear from the context. For a non-decreasing function $f$, let $f^{\leftarrow}$ denote the left-continuous inverse of $f$.
The abbreviation $a\sim b$ stands for $a$ is asymptotically equivalent to $b$.
Finally, the multiplication, division and maximum operation between vectors is meant  componentwise.
Let $\bX=(X_1,\dots, X_d)$ be  a $d$-dimensional random vector with  distribution $F_{\bX}$ and margins $F_{X_j}$, $j=1,\ldots,d$, and  $\bX_1,\bX_2,\ldots$ be independent and identically distributed (iid) copies of $\bX$.
Assume that $F_{\bX}$ is in the maximum-domain of attraction (simply domain of attraction) of a multivariate extreme-value distribution $G$, in symbols $F_{\bX}\in \MDA(G)$. This means 
that there are sequences of constants $\ba_n>\bzero=(0,\ldots,0)$  and $\bb_n\in\real^d$ such that 
$(\max(\bX_1,\ldots,\bX_n) - \bb_n)/\ba_n\rightsquigarrow \bbeta$ as $n\to \IF$, where
the distribution of $\bbeta$ is a multivariate extreme-value distribution \citep[e.g., ][pp. 147-153]{r6} of the form
\begin{equation*}\label{eq:G(x)}
G(\bx) = C_{G} \bigl( G_1(x_1), \ldots, G_d(x_d) \bigr), \qquad \bx \in \reald.
\end{equation*}
Precisely, $G_j$'s are members of the generalized extreme-value distribution (GEV) \citep[e.g., ][p. 21]{r6} and 
$C_G$ is an extreme-value copula, i.e.,
\begin{equation*}\label{eq:ev_copula}
  C_G(\bu) = \exp \bigl( - L\left(  (- \ln u_1), \ldots, (- \ln u_d)\right) \bigr), 
  \quad \bu \in (0, 1]^d,
\end{equation*}
where $L: [0,\infty)^d \mapsto [0,\infty)$ is the so-called stable-tail dependence function \citep[e.g, ][pp. 177--179]{r6}.
$G$ is a max-stable distribution, i.e., for $k=1,2,\ldots$, there are norming sequences $\ba_k'>\bzero$ and $\bb_k'\in\real^d$ such that $G^k(\ba_k'\bx+\bb_k')=G(\bx)$, for all $\bx\in\real^d$. 
The extreme-value copula expresses the dependence among extremes. 
%
%
%
%
Examples of parametric extreme-value copula models are: the Logistic or Gumbel,  the H\"{u}sler-Reiss and the Extremal-$t$, just to name a few. An extensive list of additional models is available in \citet[][Ch. 4]{r7}.
Since $L$ is a homogeneous function of order $1$, it can be conveniently represented as 
\begin{equation}\label{eq:pickA}
L(\bz)=(z_1 + \cdots + z_d) \, A( \bt ),\quad  \bz \in [0, \infty)^d,
\end{equation}
where $t_j = z_j / (z_1 + \cdots + z_d)$ for $j = 1, \ldots, d$.
The function
$A$, named Pickands (dependence) function, denotes the restriction of $L$ on the $d$-dimensional unit simplex $\simp := \left\{ (v_1,\ldots, v_d) \in [0,1]^{d}: v_1+\cdots+v_d = 1 \right\}$.
It summarizes the extremal dependence among the components of $\bbeta$, specifically it holds that $1/d\leq \max(t_1,\ldots,t_d)$ $\leq A(\bt)\leq 1$, where the lower and upper bounds represent the cases of complete dependence and independence.
A synthesis of the extremal dependence is provided by
the extremal coefficient, that is
$
\theta(G)=dA(1/d,\ldots,1/d)\in [1,d].
$
It can be 
interpreted as the (fractional) number of independent variables in a $d$-dimensional random vector with joint distribution $G$ and common margins.
An alternative summary index that measures the dependence among observations falling in the upper tail region is the coefficient of upper tail dependence \citep[e.g., ][Ch. 2.13]{r7}.
In the bivariate case it is equal to
$$
\lambda(F_{\bX})= \lim_{u\uparrow 1} \prob(X_1> F_{X_1}^{-1}(u) \lvert  X_2> F_{X_2}^{-1}(u))
=
\lim_{u\uparrow  1 } \prob(X_2> F_{X_2}^{-1}(u) \lvert  X_1> F_{X_1}^{-1}(u))
\in [0,1].
$$
It is said that $F_{\bX}$ exhibits independence or dependence in the upper tail whenever 
$\lambda(F_{\bX})=0$ or $\lambda(F_{\bX})>0$, respectively, with the case of complete dependence covered when $\lambda(F_{\bX})=1$.
 The coefficient $\lambda(F_{\bX})$ is linked to the extremal coefficient by the relationship $\theta(G)=2-\lambda(F_{\bX})$.

Nowadays, applications involving complex phenomena frequently deal with the analysis of aggregated data. 
This is especially true in big-data problems, where it may be convenient (or unavoidable) to work with aggregated data in order to reduce the computational cost.
Examples of aggregated data are the total amounts and maximum amounts obtained on a random number of observations.
Such aggregated data are realizations of the random vectors
\begin{equation} \label{eq:rnd_max}
\bS_N=\left(\sum_{i=1}^N X_{i,1}, \ldots, \sum_{i=1}^N X_{i,d}\right),\qquad
\bM_N=\left(\max_{1 \le i \le N} X_{i,1}, \ldots, \max_{1 \le i \le N} X_{i,d}\right),
\end{equation}
where $N$ is a discrete random variable defined on
$\N_+=\N\backslash\{0\}$. Assume hereafter that $N$ with distribution $F_N$ is independent of $\bX_i$'s.
For dimension $d=1$ the tail behaviour of $S_N$ has been extensively studied in the literature \citep[e.g., ][]{r8,r9}
and only few results are known on the extremes of $M_N$ \citep{r11, r12}.

The first main purpose of this contribution is to extend the classical probabilistic theory on the extreme-values to the case when the latter are computed using 
replicates of $\bM_N$ (the random vector that represents aggregated data). 
In Section \ref{sec:app} we illustrate an application, in the context of big data on Internet traffic, which would benefit from such new theoretical developments.
In the literature there are no results which can point to how different the extremal behaviour of $F_{\bM_N}$ is with respect to $F_{\bX}$, 
some preliminary findings are available in  \citet{r10}.
In this work, with Theorem \ref{theo:new_doa}, we provide the
attractor for the joint distribution of $(\bM_N,N)$, where from the extremal behaviour of $\bM_N$ (our main focus) can be deduced, such as the tail behaviours of $F_{\bM_N}$.

Some more interesting results are obtained in the case where $F_N$ is very heavy-tailed, which appears to be the more relevant one in a context of big-data, as more data are produced and aggregating them is beneficial.
Specifically, when $\prob(N> y) =y^{-\alpha} \slowf(y)$, with $y>0$ and $\alpha\in(0,1)$, where $\slowf$ is a slowly varying function at infinity, and $F_{\bX}\in \MDA(G)$, 
we show that $\bM_N$ and $N$ are asymptotically dependent. 
Furthermore,  we find that in this case
$F_{\bM_N }\in \MDA(G_\alpha),$ where $G_\alpha$ is a new max-stable distribution with an extreme-value copula
$C_{G_\alpha}$, given in \eqref{eq:ev_alpha_copula}, differing from $C_G$, the extreme-value copula of $G$. 
The coefficient $\alpha\in(0,1)$ influences the extremal dependence structure of $G_\alpha$.
A practical implication of our finding is that the extreme properties of $F_{\bM_N}$ can be recovered by knowing the extremal properties 
of $F_{\bX}$ and the tail behaviour of $N$. Here are two examples. 
The joint upper-tail probability of $\bM_N$ can be approximated as
$$
\prob\left(F_{M^{(1)}_{N}}\left(M^{(1)}_N\right)>1-\frac{1}{y_1} \text{ or} \ldots \text{or } F_{M^{(d)}_{N}}\left(M^{(d)}_N\right)>1-\frac{1}{y_d} \right)\sim L^{\alpha}({\by}^{-1/\alpha}),
$$
for a large enough $\by>\bzero$, where $M^{(j)}_{N}=\max_{1 \le i \le N} X_{j,1}$, $j=1,\ldots,d$ and $L$ is the stable-tail dependence function of $C_G$ (see Proposition \ref{sec:RNM_tail}). 
This means that by combining $L$ and $\alpha$, we can approximate the probability that at least one component among $M^{(1)}_{N},\ldots, M^{(d)}_{N}$ exceedes a high percentile of its own distribution.
When $d=2$, we also have
\begin{equation}\label{lyke}
\lambda(F_{\bM_N})  = 2- [2- \lambda\left( F_{\bX} \right) ]^{\alpha}.
\end{equation}

The second purpose of this contribution concerns an inverse statistical problem in the context of extremes of aggregated data, which also motivates the study of the asymptotic joint distribution of $(\bM_N,N)$.
Precisely, in applications where the variables ${\bM_N}$ and $N$ are observable in place of $\bX$, 
the interest may however be in inferring the extremal dependence of the distribution $G$. 
Theorem \ref{theo:new_doa} provides the mathematical ground to address this issue,
as it gives a joint (limiting) statistical model for the sample extremes of $\bM_N$ and $N$.
In particular, the sample maxima of the latter random vector and variable can be used to obtain estimators of  $\lambda\left(F_{\bM_N}\right)$ and $\alpha$, respectively. Then, 
by exploiting the relation among extremal properties of $F_{\bM_N}$ and $F_{\bX}$ and the tail behaviour of $N$, and solving \eqref{lyke} for $\lambda(F_{\bX})$, an estimator of the latter can be obtained.
More generally, we  focus on the Pickands function as it allows to derive the $\lambda$ and $\theta$ coefficients and other related quantities. 
We define a new semi-parametric procedure for inferring $A$, which combines together preliminary estimators for $\alpha$ and $A_\alpha$ (the Pickands function relative to $C_{G_\alpha}$).
Specifically, we consider a likelihood- and moments-based estimator for $\alpha$ and we use three nonparametric estimators for $A_\alpha$ (existing in the literature). 

The rest  of the paper is organized as follows. In Section \ref{sec:RNM}, we present our main theorem providing the attractor for the joint distribution of $(\bM_N,N)$.
Different representations for the distribution $G_\alpha$ are derived 
in Section \ref{ssec:randsc}. 
In Section \ref{sec:inversion}, by an inversion method, we define an estimator for inferring $A$.
We establish its asymptotic properties (Theorem \ref{theo:conv_estimators}) and show its finite-sample performance by a simulation study.
Finally, we discuss directions for future research in Section \ref{sec:app}, including a real data example which appears as a promising field of application for  our theoretical framework. 
The proofs are reported in the Appendix, whereas some technical details and additional simulation results are included in the supplementary  material.

\section{Main Results}\label{sec:RNM} 

First, recall that the members of the GEV distribution are: 
the $\alpha$-Fr\'{e}chet (heavy-tailed distribution),
Gumbel (light-tailed distribution) and Weibull (short-tailed distribution),
in symbols, $\Phi_\alpha(x)=\exp(-x^{-\alpha})$, with $x>0$ and $\alpha>0$, 
$\Lambda(x)=\exp(-e^{-x})$ with $x\in\real$ and $\Psi_\alpha(x)=\exp(-(-x)^{-\alpha})$
with $x<0$. In the sequel, for a positive random variable, say $V$, we denote its  Laplace transform by $\LAP_V(s)=\EE{e^{-sV}}$, $s>0$.
We also recall that a random variable, $S$, is positive (asymmetric) $\alpha$-stable with
index parameter $0<\alpha<1$ if its Laplace transform is $\LAP_S(s)=e^{-s^{-\alpha}}$.

Let $N$ be a random block size and $\bM_N$ be a vector of componentwise maxima obtained
with a randomly sized block of iid random vectors $\bX_1,\bX_2,\ldots$ with common distribution  $F_{\bX}$, defined in \eqref{eq:rnd_max}.
Note that
\begin{eqnarray*}
F_{\bM_N}(\bx)=\sum_{n=1}^{\infty}F_{\bX}^n(\bx)\prob(N=n)=\sum_{n=1}^{\infty}\exp[-n\{-\ln F_{\bX}(\bx)\}]\prob(N=n)=\LAP_N(-\ln F_{\bX}(\bx)).
\end{eqnarray*}
Assuming that $F_{\bX}\in \MDA(G)$ and $F_N\in \MDA(H)$, where either $H\equiv \Phi_\alpha$ or $H\equiv \Lambda$ since $N$ is positive integer-valued \citep{r9}, we establish new limit results 
concerning the attractor for the joint distribution of the random vector $(\bM_N, N)$ and the
tail behaviour of the random vector $\bM_N$.
\subsection{Domains of Attraction}\label{sec:RNM_Doa}
The first limit result provides the attractor for the joint distribution of the random vector $(\bM_N, N)$. 
\begin{theo}\label{theo:new_doa} 
Assume that $F_{\bX}\in \MDA(G)$ and $F_N\in \MDA(H)$ with $H\equiv \Phi_\alpha$ or $H\equiv \Lambda$. Then, there exist  
norming constants $\bc_n>\bzero, \kappa_n>0$, $\bd_n\in \real^d, \varrho_n \inr$ such that
$$
\limit{n}\prob^n\left(\frac{\bM_N-\bd_n}{\bc_n} \le \bx, \frac{N-\varrho_n}{\kappa_n} \le y\right)= Q(\bx, y),
$$ 
where $Q$ is defined as follows: 
%
	%
	\begin{enumerate}
		\item if $F_N\in\MDA(\Phi_\alpha)$, then 
		\begin{equation}\label{eq:joint_limiting}
		-\ln Q(\bx,y)=
		\begin{cases}
		y^{-\alpha}e^{-y\sigma(\bx;\alpha)}+ 
		\sigma^\alpha(\bx,\alpha)
		\,\gamma(
		1-\alpha,y\,\sigma(\bx,\alpha)), & \alpha\in(0,1)\\
		-\ln G(\bx) + y^{-\alpha} 
		, & \alpha\ge 1
		\end{cases}
		\end{equation}
		for all $\bx\in\real^d$ and $y>0$, where $\sigma(\bx,\alpha)=(-\ln G(\bx))/\Gamma^{1/\alpha}(1-\alpha)$ 
		and $\Gamma$, $\gamma$ denote the Euler Gamma and  Lower Incomplete Gamma functions, respectively. 
		%
		\item if $F_N\in\MDA(\Lambda)$, then
		\begin{equation}\label{eq:joint_limiting_gumb}
		-\ln Q(\bx,y)= -\ln G(\bx)+e^{-y}, \quad \bx\in\real^d,\; y\in\real.
		\end{equation}
	\end{enumerate}
\end{theo}
The margins of $Q$ are
\begin{equation}\label{eq:margins}
\begin{cases}
G_{\alpha}(\bx):= \exp( -  (- \ln G(\bx))^\alpha), & \alpha\in(0,1)\\
G(\bx), & \alpha\geq 1
\end{cases}
\end{equation}
and $\Phi_\alpha(y)$, when $F_N\in\MDA(\Phi_\alpha)$, $\alpha>0$. 
While they
are equal to $G(\bx)$, $\bx\in\real^d$, 
and $\Lambda(y)$, $y\in\real$, when $F_N\in\MDA(\Lambda)$.
Specifically, the distribution $G_\alpha$, $\alpha\in(0,1)$, is  max-stable  
with margins $G_{\alpha,1},\ldots,G_{\alpha,d}$, that are members of the GEV class, and extreme-value copula 
%
%
\begin{equation}\label{eq:ev_alpha_copula}
C_{G_\alpha}(\bu) = \exp \bigl( - 
L^\alpha\bigl((-\ln u_1)^{1/\alpha}, \ldots, (-\ln u_d)^{1/\alpha}\bigr)\bigr), \quad \bu \in (0, 1]^d, \quad \alpha\in(0,1),
\end{equation}
where $L$ is the stable-tail dependence function of the max-stable distribution $G$. 

A probabilistic interpretation of the problem addressed in Theorem \ref{theo:new_doa} is as follows.
Let $N_1,N_2,\ldots$ be iid copies of $N$ and set
$
N_n^{\scriptsize{+}}=N_1+\cdots+N_n$, $N_n^{\scriptsize{\vee}}=\max(N_1,\ldots,N_n), 
$
 then we have 
$$
\prob\left(
\frac{\bM_{N_n^{\scriptsize{+}}}-\bd_n}{\bc_n}\leq \bx,\frac{N_n^{\scriptsize{\vee}}-\varrho_n}{\kappa_n}\leq y
\right)
=\prob^n\left(\frac{\bM_N-\bd_n}{\bc_n} \le \bx, \frac{N-\varrho_n}{\kappa_n} \le y\right).
$$
Loosely speaking, we are concerned with the asymptotic distribution of the random vector $(\bM_{N_n^{\scriptsize{+}}},N_n^{\scriptsize{\vee}})$ appropriately normalized (a.n.).
When $F_N\in\MDA(\Lambda)$ (light-tailed) or 
$F_N\in\MDA(\Phi_\alpha)$ (heavy-tailed), with $\alpha>1$, then $\mu=\expect(N)<\infty$. Since
$n^{-1}N_n^{\scriptsize{+}}$ converges to $\mu$, then the asymptotic distributions of $\bM_{N_n^{\scriptsize{+}}}$ and
$\bM_{\floor*{n\mu}}$ a.n. are approximately the same. %
When $F_N\in\MDA(\Phi_\alpha)$ (heavy-tailed), with $ \alpha \in (0,1)$, then 
 $\tilde{\kappa}_n^{-1}N_n^{\scriptsize{+}}$ converges in distribution to a positive stable random variable $S$, where $\tilde{\kappa}_n:= \kappa_n\Gamma^{1/\alpha}(1-\alpha)$. Consequently, the asymptotic distributions of $\bM_{N_n^{\scriptsize{+}}}$ and $\bM_{\floor*{\tilde{\kappa}_n S}}$ a.n. coincide 
and are equal to $G_\alpha$ in the first line of \eqref{eq:margins}, see Corollary \ref{corr:loc_sc_mix}, which is a location-scale mixture of the limiting max-stable distribution $G$, obtained with a deterministic block size. A similar result can be established in the case of $\alpha=1$, which is not explicitly discussed here for the sake of brevity.
\begin{korr}\label{corr:loc_sc_mix}
Let $F_{\bX}\in \MDA(G)$, $F_N\in\MDA(\Phi_\alpha)$, $\alpha\in(0,1)$ and $\bc_n$, $\bd_n$ and $\kappa_n$ as in Theorem \ref{theo:new_doa}.
Then $\tilde{\kappa}_n^{-1}N_n^{\scriptsize{+}} \rightsquigarrow S$, as $n \to \infty$, where $\tilde{\kappa}_n=\kappa_n\Gamma^{1/\alpha}(1-\alpha)$  and 
$$
\lim_{n \to \infty}\prob(\bM_{\floor{\tilde{\kappa}_n S}} \leq \bc_n \bx + \bd_n) = G_\alpha(\bx).
$$
\end{korr}
As for asymptotic dependence, we point out the following. When $F_N \in \MDA (\Phi_\alpha)$, $\alpha \geq 2$, or $F_N\in\MDA(\Lambda)$, $N_n^{\scriptsize{+}}$ and $N_n^{\scriptsize{\vee}}$ a.n. converge to nondegenerate independent random variables \citep[e.g.][pp. 1-2]{r0}. Intuitively, this explains the asymptotic independence between $\bM_{N_n^{\scriptsize{+}}}$ and $N_n^{\scriptsize{\vee}}$ (and, in turn, the extremal independence between $M_N$ and $N$). In the special case of $F_N \in \MDA (\Phi_\alpha)$, $\alpha \in(0,2)$, $N_n^{\scriptsize{+}}$ and $N_n^{\scriptsize{\vee}}$ a.n. are asymptotically dependent \citep[e.g.][pp. 1-2]{r0}. 
Thus, when $\alpha \in [1,2)$ we explain the asymptotic independence between $\bM_{N_n^{\scriptsize{+}}}$ and $N_n^{\scriptsize{\vee}}$ a.n. in a different way. From the derivations in Lemmas \ref{lem: tail ratio}-\ref{lem: lim dep} and Section \ref{sub:second_theo_main} one sees that the dependence structure of the limiting distribution $Q(\bx,y)$ is determined by the limit of the conditional exceedance probability 
$
1-\prob(\bM_{N}\leq \bc_n \bx +\bd_n|N>\kappa_n y+\varrho_n)
$
, as $n \to \infty$. 
Note that, as $n$ grows, the norming sequences $(\bc_n, \bd_n)$ and $(\kappa_n, \rho_n)$ affect this threshold exceedance probability in two opposing ways: the first ones increase the threshold $\bc_n\bx+\bd_n$; the second ones force $N$, and thus $\bM_N$, to be stocastically larger and larger. Clearly, a non-degenerate limit -- i.e. in $(0,1)$ -- is obtained only if the two effects offset each other.

On one hand $$F_{\bM_{N_n^{\scriptsize{+}}}}(\bc_n \bx + \bd_n)=\LAP^n_N(-\ln F_\bX(\bc_n \bx + \bd_n))$$
 and so $\bc_n$ and $\bd_n$ are affected by the behaviour of $1-\LAP_N(1/y)$ as $y \to \infty$. On the other hand $F_{N_n^{\scriptsize{\vee}}}(\kappa_ny+\varrho_n)=F^n_N(\kappa_ny+\varrho_n)$ and so $\kappa_n$ and $\varrho_n$ depend on the tail properties of $F_N$. When $\alpha \in (0,1)$, $\{1-F_N(y)\}/\{1-\LAP_N(1/y))\}$ converges to a nondegenerate limit as $y \to \infty$, while it converges to zero when $1\leq \alpha<2$. Accordingly, in the first case, the limit of the conditional probability of exceedances is positive, whereas it is zero in the second case, since $\bc_n$, $\bd_n$ are ``too heavy" relative to $\kappa_n$, $\varrho_n$. As a result, the marginal distributions of $\bM_{N_n^{\scriptsize{+}}}$ and $N_n^{\scriptsize{\vee}}$ a.n. converge to non-degenerate limits but their dependence is wiped out as $n \to \infty$.

The dependence structure of the ($d$+1)-dimensional distribution $Q$  defined in \eqref{eq:joint_limiting} and \eqref{eq:joint_limiting_gumb} can be synthesized by means of its extremal coefficient.
\begin{korr}\label{corr:extremalcoeff}
When the expression of $Q$ is given
in \eqref{eq:joint_limiting}, then the extremal coefficient is 
$$
\theta(Q)\equiv \theta(G, \alpha):=
\begin{cases}
1-\mathcal{E}_{\Gamma^{1/\alpha}(1-\alpha)}(\theta(G))+(\theta(G))^\alpha \mathcal{G}_{1-\alpha, \Gamma^{1/\alpha}(1-\alpha)}(\theta(G)), & \alpha\in(0,1)\\
\theta(G) + 1, & \alpha \geq 1.
\end{cases},
$$
where  $\mathcal{E}_b(\cdot)$ and $\mathcal{G}_{a, b}(\cdot)$ denote the Exponential and Gamma cumulative distribution functions, with shape and scale parameters $a$ and $b$. In particular, for $\alpha \in (0,1)$ we have 
\begin{equation}\label{eq:excoefprop}
1\leq \theta(G, \alpha) \leq 1+\theta(G), \qquad \lim_{\alpha\to 1^-}\theta(G,\alpha)=1+\theta(G).
\end{equation}
When the expression of $Q$ is given
in \eqref{eq:joint_limiting_gumb}, $\theta(Q)=\theta(G) + 1$.
\end{korr}
From the first result in the left-hand side of \eqref{eq:excoefprop}  
we deduce that the extremal coefficient of $Q$ is larger than or equal to 1 (as expected) and bounded from above by $1+\theta(G)$, representing the case where $\bM_N$ and $N$ have no tail dependence. Moreover, the second result in the right-hand side of \eqref{eq:excoefprop} highlights a continuous transition of the extremal dependence level from asymptotic dependence between $\bM_N$ and $N$ (i.e. $\alpha \in (0,1)$) to asymptotic independence (i.e. $\alpha \geq 1$), when $F_N$ belongs to the $\alpha$-Fr\'echet domain. 

The extremal coefficient for the distribution $G_\alpha$ with $\alpha\in(0,1)$, 
$\theta(G_\alpha)$, is given in \eqref{eq:new_extremal_coeff}. 
In the bivariate case, by \eqref{eq:new_extremal_coeff} and
the relationship between the extremal coefficient and the coefficient of upper tail dependence, i.e., $\theta(G)=2-\lambda(F_{\bX})$, 
we obtain
$$
\theta(G_\alpha)=[2-\lambda(F_{\bX})]^\alpha.
$$
Since it is also true that $\theta(G_\alpha)=2-\lambda(F_{\bM_N} )$,
we obtain the result \eqref{lyke}.

\subsection{Tail Behaviours}\label{sec:RNM_tail}
The second limit result establishes the tail behaviours of the random vector $\bM_N$, 
which concerns the probability that at least one component of
the random vector $\bM_N$ exceeds an increasingly large value. 
\def\gs{\widetilde{G}}

In the sequel, for a given max-stable distribution $G$ we denote by $\gs $ a distribution with the same copula as $G$ and common unit-Fr\'{e}chet margins.
\begin{sat}\label{prop:tail_behaviours}
Assume that $F_{\bX}\in \MDA(G)$ and $F_N\in \MDA(H)$ with $H\equiv \Phi_\alpha$ or $H\equiv \Lambda$. For $\by\in (0,\infty)^d$ and $n\in\N_+$,
\begin{enumerate}
\item if $F_N\in\MDA(\Phi_\alpha)$ with $0<\alpha\leq 1$, then
\begin{eqnarray*}
1 - \prob\left(\bM_N\leq U_{\bX}(n\by)\right)&\sim&
\begin{cases}
	 \Gamma(1-\alpha)\prob\left(N>\frac{n}{-\ln \gs(\by)}\right), \hspace{1.3em} \alpha \in (0,1)\\
	\left\{1-\LAP_N(1/n)\right\}\{-\ln \gs(\by)\},\hspace{1.3em} \alpha=1
\end{cases},
\quad n \to \infty\\
1 - \prob\left(\bM_N\leq U_{\bM_N}(n\by)\right)
&\sim& \left\{-\ln  \gs \left({\by}^{1/\alpha}\right)\right\}^\alpha n^{-1},\hspace{10.5em}  n \to \infty
\end{eqnarray*}
\item if $F_N\in\MDA(\Phi_\alpha)$ with $\alpha >1$ or $F_N\in\MDA(\Lambda)$, then
\begin{eqnarray*}
1 - \prob\left(\bM_N\leq U_{\bX}(n\by)\right)&\sim& 
n^{-1}\EE N\{-\ln  \gs (\by)\}
,\quad n\to\infty\\
1 - \prob\left(\bM_N\leq U_{\bM_N}(n\by)\right)
&\sim& n^{-1}\{-\ln  \gs (\by)\},\hspace{3.3em} n\to\infty
\end{eqnarray*}
\end{enumerate}
where $U_{\bX}(n\by),U_{\bM_N}(n\by) \to \binf$ as $n\to\infty$, with $U_{\bX}$ and $U_{\bM_N}$ defined in Appendix \ref{sub:prop_tail}.
\end{sat}
Set $p_j=(ny_j)^{-1}$, $j=1,\ldots,d$, and recall that $L$ denotes the stable-tail dependence function of $G$.
As $n\to\binf$, by Proposition \ref{prop:tail_behaviours}, the probability that
at least one component $\bM_N^{(j)}$ of $\bM_{N}$ exceeds the $1-p_j$ quantile of its own distribution is approximately $L((n\by)^{-1})$,  when $\EE N <\infty$,  while it is approximately $L^\alpha(({n\by})^{-1/\alpha})$, when $\EE{ N} =\infty$.

\section{Representations of the model $G_\alpha$} \label{ssec:randsc}

In this section we show that there are different constructions that yield a max-stable 
distribution with the same copula $C_{G_\alpha}$ in \eqref{eq:ev_alpha_copula} of the
distribution $G_\alpha$. 
Furthermore, we derive the  Pickands function corresponding to $C_{G_\alpha}$.

Let $S$ be a positive $\alpha$-stable random variable with index parameter $0<\alpha<1$.
Let $\bZ$ be a random vector with max-stable distribution $\gs$. 
Assume  $S$ and $\bZ$ to be independent.
Define  $\bR=(SZ_1,\ldots,S Z_d)$, then for every $\by>\bzero$, 
\begin{equation}\label{eq:scaledfre}
\prob(\bR \le \by)=\EE{ \gs^{S}\left(\by\right)}=
 \LAP_S\left(  -\ln \gs\left(\by\right)\right)
 = \exp\left( -\left(-\ln \gs \left(\by\right)\right)^{\alpha}\right)=:\gs_{\alpha}(\by).
\end{equation}
The distribution $\gs_\alpha$ is a special case of $G_\alpha$, that is max-stable with extreme-value copula $C_{G_\alpha}$ and common $\alpha$-\FRE margins.
By \eqref{eq:scaledfre}, it follows easily that for any $\alpha_1, \alpha_2 \in (0,1)$ we have 
$
(\gs_{\alpha_1})_{ \alpha_2}= ( \gs_{\alpha_2})_{\alpha_1}= \gs_{ \alpha_1\alpha_2}.
$
This means that the max-stable $\widetilde{G}_\alpha$ can be obtained by (iterated) random scaling in various ways.
In the particular case that the components of $\bZ$ are independent, the copula of $\gs_{\alpha}$ is
\begin{equation}\label{eq:logistic_cop}
C_{\gs_{\alpha}}(\bu)= \exp\Bigl( -\Bigl(  (-\ln u_1)^{1/\alpha}+\cdots +(-\ln u_1)^{1/\alpha}\Bigr)^\alpha\Bigl),\quad \bu \in (0, 1]^d, \quad \alpha\in(0,1),
\end{equation}
which is the well-know Symmetric Logistic copula \citep[e.g., ][p. 172]{r7}. 
Therefore, the elements of $\bZ$ are dependent for any $\alpha\in(0,1)$, and they become  nearly independent as $\alpha\rightarrow 1$ and almost completely dependent as $\alpha\rightarrow 0$. Random scaling constructions similar to this one have been already discussed by \cite{r49} and \cite{r50}.

The de Haan representation of max-stable processes \citep{r13} provides
a Poisson point process
construction of a random vector with any max-stable distribution $\gs$.
A 
question that arises here is: What is the spectral representation of 
a random vector $\bR$ defined by the random scaling construction?
The next result establishes that the findings presented in \cite{r14} indeed provide 
the spectral representation of $\bR$. 
\begin{sat}\label{prop:ppmax_max}
Let $\bZ_1,\bZ_2,\ldots$ be iid copies of $\bZ$, with distribution $\gs$, independent 
of $P_1,P_2,\ldots$ that are points of a Poisson process on $(0,\IF)$ with intensity measure 
$\alpha r^{-(\alpha+1)}dr$, $\alpha\in(0,1)$. Define
$$
\bR = \frac 1 {\Gamma(1-\alpha)}\Bigl( \max_{i\ge 1} P_i Z_{i1} \ldot 
\max_{i\ge 1} P_i Z_{id}  \Bigr).
$$
Then, the distribution of $\bR$ is $\gs_{\alpha}$.
\end{sat}
Here we provide an alternative, more general proof than that given in \cite{r14}. Specifically, ours does not rely on an unnecessary smoothness assumption.
Next, we provide a characterization of $G_\alpha$, as the attractor distribution for a general random scaling and centering construction.
\begin{sat}\label{pro:doa_rnd:scale}
Let $\bX_1, \ldots, \bX_n$ be iid copies of the random vector $\bX$, with distribution $F_{\bX}$. Assume $F_{\bX}\in \MDA(G)$. Let $S$ be a positive $\alpha$-stable random variable, $\alpha\in(0,1)$. Assume $S$ is independent of $\bX$. Define
$$
\bw_n := \frac{\ba_{\floor{nS}}}{\ba_n}, \quad \bv_n:=\bb_n- \bb_{\floor{nS}} \frac{\ba_n}{\ba_{\floor{nS}}},
$$
where $\ba_n$ and $\bb_n$ are the usual norming constants of $F_\bX$. Then,
$$
\ba_n^{-1}\left(\max( \bw_n(\bX_1 -\bv_n), \ldots, \bw_n(\bX_n -\bv_n))- \bb_n \right) \rightsquigarrow G_\alpha, \quad n \to \infty.
$$
\end{sat}
A simple implication of Proposition \ref{pro:doa_rnd:scale} is the following. 
Transforming $\bX$ into $\bY$, a random vector with common unit-Fr\'{e}chet marginal distributions, and setting $\bR=S\bY$, implies that $F_{\bR} \in \mathcal{D}(\gs_{\alpha})$. To see this, note that $F_\bY \in \mathcal{D}(\widetilde{G})$ with norming sequences $\ba_n=\bone n$, $\bb_n =\bzero$, and that, as $n\to \infty$, $S\bY$ and $\bw_n \bY=(\floor{nS}/n)\bY$ have approximately the same distribution.
As a result, the attractors of $F_{\bR}$ and $F_{\bM_N}$, when $F_N  \in \mathcal{D}(\Phi_\alpha)$, $\alpha \in (0,1)$, share the same extreme-value copula, $C_{G_\alpha}$.
We finally derive  the explicit form of the Pickands function corresponding to the latter.
\begin{sat}\label{pro:pickands_rnd:scale}
The Pickands function corresponding to the extreme-value copula $C_{G_\alpha}$ in \eqref{eq:ev_alpha_copula} is
\BQN \label{eq:new_pick}
A_\alpha(\bt)=\norm{\bt}_{1/\alpha} A^\alpha \left(\left(\bt /\norm{\bt}_{1/\alpha}\right)^{1/\alpha}\right),\quad \bt \in \simp,\;\alpha\in(0,1),
\EQN
where $A$ is the Pickands dependence function corresponding to $C_G$ and 
\begin{equation} \label{eq:old_pick}
\norm{\bt}_{1/\alpha}=\Biggl(\sum_{i=1}^{d} t_i^{1/\alpha} \Biggr)^{\alpha}, \quad \bt \in \simp, \; \alpha\in(0,1)
\end{equation}
is the Pickands function corresponding to the Logistic copula.
\end{sat}
As a direct consequence of Proposition \ref{pro:pickands_rnd:scale}, the following facts ensue.
The smaller the parameter $\alpha$, the more $A_\alpha$ represents a stronger dependence level than $A$.
Since we have that $\norm{(1/d,\ldots,1/d)}_\alpha=d^{\alpha-1}$, then, by the definition of the extremal coefficient in Section \ref{sec:intro}, we obtain
\begin{equation}\label{eq:new_extremal_coeff}
\theta(G_\alpha)=(\theta(G))^\alpha.
\end{equation}
By solving for $A$ in equation \eqref{eq:new_pick}, we obtain the inverse relation between $A_\alpha$ and $A$, i.e., 
\BQN \label{eq:inverted_pick}
A^{\star}(\bt):=
A\left(\left(\bt/ \norm{\bt}_{1/\alpha}\right)^{1/\alpha}\right) = 
\left(A_\alpha(\bt)/\norm{\bt}_{1/\alpha}\right)^{1/\alpha},\quad \bt\in\simp
\EQN
and 
$A(\bt)=A^{\star}(\bt^{\alpha}/ \norm{\bt^{\alpha}}_1),$
 providing the expression for the Pickands function in \eqref{eq:pickA}.
\section{Inferring the Pickands function}\label{sec:inversion}

In this section we introduce a new semiparametric procedure to estimate the Pickands function $A$ in \eqref{eq:pickA}. Several nonparametric estimators are already available for $A$ when a data sample from the limiting distribution for unaggregated data, $G$, is observable, see e.g., \cite{r46}, \cite{r45}, \cite{r16}, \cite{r24}, \cite{r17} among others. 
Unlike the above references, we assume that only replicates of $(\bM_N,N)$ are observable, from which sample extremes (maxima) are extracted. Then, we construct an estimator for $A$, exploiting an inversion method via \eqref{eq:inverted_pick}. This is a substantial novelty in the extreme-value literature. 
It is common practice to assume that sample maxima are exactly coming from the limiting model $Q$ in \eqref{eq:joint_limiting}, with $\alpha\in (0,1)$, and provide an asymptotic validation of the proposed inferential procedure in such a setting.
Via an extensive simulation study in Section \ref{sec:simulation}, we show that, in practice, our method provides a good performance with data that are only approximately coming from $Q$.
Extending the asymptotic statistical theory to the latter case goes beyond the scope of the present already quite technical work.
Observe that, since $\bt\mapsto(\bt/\norm{\bt}_\alpha)^{1/\alpha}$ is a bijective map,
estimating $A^\star$ is equivalent to estimating $A$, thus, for simplicity, we hereafter focus on the former function. 

\subsection{A semiparametric composite-estimator}\label{sec:nonparest}
Let $(\bbeta_1, \xi_1),(\bbeta_2, \xi_2),\ldots,$, be iid random vectors with joint distribution in \eqref{eq:joint_limiting} with $\alpha \in (0,1)$. 
Assume that a sample of $n$ observations from such a sequence is available. An estimate of $A^\star$ is obtained by combining the results of a two-step procedure: we estimate $\alpha$ and
$A_\alpha$, we plug the estimates in \eqref{eq:inverted_pick}.
Precisely, $\xi_1,\ldots,\xi_n$ follows a $\alpha$-Fr\'{e}chet distribution. For estimating
$\alpha$ we consider two well-known estimators: the Generalized Probability Weighted Moment (GPWM) \citep{r18} and the Maximum Likelihood (ML). 
In the first case the estimator is
\begin{equation}\label{eq:gpwm}
\estalpha^{\GPWM}:=
\left(k-2\frac{\widehat{\mu}_{1,k}}{\widehat{\mu}_{1,k-1}}\right)^{-1},
\end{equation}
for $k \in \N_+$, where
$$
\widehat{\mu}_{a,b}=\int_0^1 H_n^{\leftarrow}(v) v^{a}(-\ln v)^b \diff v, \quad a,b\in\N
$$
and
\begin{equation}\label{eq:emp_margin}
H_n(y)=\frac{1}{n}\sum_{i=1}^n \indic(\xi_i\leq y), \quad y>0.
\end{equation}
In the second case the estimator is
\begin{equation}\label{eq:mle}
\estalpha^{\ML}:= \argmax_{\tilde{\alpha}\in(0,\infty)}\sum_{i=1}^n \ln \dot{\Phi}_{\tilde{\alpha}}(\xi_i),
\end{equation}
where $\dot{\Phi}_{\tilde{\alpha}}(x)=\partial/\partial x \, \Phi_{\tilde{\alpha}}(x)$, $x>0$.

The sequence $\bbeta_1,\ldots\bbeta_n$ follows the distribution $G_\alpha$. For estimating
$\picka$ we consider three well-know estimators: Pickands (P) \citep{r19}, 
Cap\'{e}ra\`{a}-Foug\`{e}re-Genest (CFG) \citep{r20} and Madogram (MD) \citep{r17}. In the first case the estimator is
\begin{eqnarray}\label{eq:pick_est_pick}
\estpicka^\PICK(\bt)&:=&\left(\frac{1}{n}\sum_{i=1}^n \estlambda_i(\bt)\right)^{-1},\\
\nonumber \estlambda_i(\bt)&=&\min_{1\leq j\leq d}\left\{-\frac{1}{t_j}\ln\left(\frac{n}{n+1}G_{n,j}(\eta_{i,j})\right)\right\}
\end{eqnarray}
%
%
%
where for every $x\in \real$ and $j\in\{1,\ldots,d\}$
\begin{equation}\label{eq:marg_empfun}
G_{n,j}(x)=\frac{1}{n}\sum_{i=1}^n\indic(\eta_{i,j}\leq x).
\end{equation}
In the second case the estimator is 
\begin{equation}\label{eq:pick_est_cfg}
\estpicka^\CFG(\bt):=\exp\left(-\frac{1}{n}\sum_{i=1}^n \ln\estlambda_i(\bt)-\varsigma\right),
\end{equation}
where $\varsigma$ is the Euler's constant. Finally, in the third case the estimator is
\begin{eqnarray} 
\label{eq:pick_est_md} \estpicka^\MD(\bt)&:=&\frac{\estmado(\bt) + c(\bt)}{1 - \estmado(\bt) - c(\bt)},\\
\label{eq:pick_mado} \label{madoest}\estmado(\bt) &=& \frac{1}{n}\sum_{i=1}^n 
  \left(
    \max_{j=1,\ldots,d} G_{n,j}^{1/t_j}(\eta_{i,j})  -
    \frac{1}{d}\sum_{j=1}^d   G_{n,j}^{1/t_j}(\eta_{i,j})  
  \right),
\end{eqnarray}
where $u^{1/0}=0$ for $0<u<1$ by convention and $c(\bt)=\sum_{j=1}^d t_j/(1+t_j)$.
%
%

For brevity we denote the estimators of $\alpha$ and $A_\alpha$ by  $\estalpha^{\scriptscriptstyle{\bullet}}$ 
and $\estpicka^{\scriptscriptstyle{\circ}}$, respectively, where the symbols ``$\scriptstyle{\bullet}$" and ``$\scriptstyle{\circ}$" are representative
of the labels ``GPWM", ``ML" and ``P", ``CFG", ``MD", respectively, Then, plugging the estimators into equation \eqref{eq:inverted_pick}
we obtain the following composite-estimator for $\pickainv$,
\begin{equation}\label{eq:est_inv_pick}
\estpickin^{\scriptscriptstyle{\circ,\bullet}}(\bt):=
\left(\estpicka^{\scriptscriptstyle{\circ}}(\bt)/\norm{\bt}_{1/\estalpha^{\scriptscriptstyle{\bullet}}}\right)^{1/\estalpha^{\scriptscriptstyle{\bullet}}},\quad \bt\in\simp.
\end{equation}
Next, we establish the asymptotic theory of the composite-estimator in \eqref{eq:est_inv_pick} 
defined by all the combinations of the GPWM and ML estimators for $\alpha$ with the P, CFG and MD estimators for $A_\alpha$. Our results rely on the following assumptions.

\begin{cond}\label{cond:cond_theo}
For $j\in\{1,\ldots,d\}$, let $\setU_j=\{\bu\in[0,1]^d: 0<u_j<1\}$. Assume that:
\begin{inparaenum}
\item \label{en:first_cond} for $j\in\{1,\ldots,d\}$, the first-order partial derivative 
$\dot{C}_{G_\alpha;j}(\bu):=\partial/\partial u_j C_{G_\alpha}(\bu)$ exists and is continuous in $\setU_j$;
\item \label{en:second_cond} for $i,j\in\{1,\ldots,d\}$, the second-order partial derivative $\ddot{C}_{G_\alpha;i,j}(\bu):=\partial/\partial u_i \dot{C}_{G_\alpha;j}(\bu)$ exists and is continuous in $\setU_i \cap \setU_j$ and
$$
\sup_{\bu\in \setU_i \cap \setU_j}\max(u_i,u_j)|\ddot{C}_{G_\alpha;i,j}(\bu)|<\infty.
$$
\end{inparaenum}
\end{cond}
\begin{theo}\label{theo:conv_estimators}
For the estimators $\estpickin^{\scriptscriptstyle{\MD,\bullet}}$, assume that Condition \ref{cond:cond_theo}\ref{en:first_cond} holds true; assume that Condition
\ref{cond:cond_theo}\ref{en:second_cond} is also satisfied for the estimators $\estpickin^{\scriptscriptstyle{\PICK,\bullet}}$, $\estpickin^{\scriptscriptstyle{\CFG,\bullet}}$. Finally, assume that the choice of $k \in \mathbb{N}_+$ in the GPWM-based estimator $\estpickin^{\scriptscriptstyle{\circ,\GPWM}}$ satisfies $\alpha>1/(k-1)$. Then, as $n\to \IF$ 
\begin{equation}\label{eq:normality_comp_est}
\sqrt{n}\left\{\estpickin^{\scriptscriptstyle{\circ,\bullet}}(\bt)-\pickainv(\bt)\right\}_{\bt\in\simp}\rightsquigarrow
\left\{(\phi_{\circ,\bullet}(\cproc_Q))(\bt) \right\}_{\bt \in \simp}
\end{equation}
in $\ell^\infty(\simp)$, for an operator $\phi_{\circ,\bullet}$ into $\ell^\infty(\simp)$ and a zero-mean Gaussian process
$\cproc_Q$, whose covariance function is
\begin{equation}\label{eq:cov_gproc}
\Cov(\cproc_Q(\bu),\cproc_Q(\bv))=C_Q(\min(\bu,\bv))-C_Q(\bu)C_Q(\bv),\quad \bu,\bv\in[0,1]^{d+1},
\end{equation}
where $C_Q$ is the copula
$$
C_Q(\bu,v)=Q(G_{\alpha,1}^{\leftarrow}(u_1),\ldots,G_{\alpha,d}^{\leftarrow}(u_d),\Phi^{\leftarrow}_\alpha(v))
$$
and the minimum is taken  componentwise. Moreover,
\begin{eqnarray*}
\|\estpickin^{\scriptscriptstyle{\circ,\bullet}}-\pickainv \|_\infty &\conP& 0,\quad 
\|\estpickin^{\scriptscriptstyle{\MD,\GPWM}}-\pickainv \|_\infty \conAS 0,\quad n\to\infty.
\end{eqnarray*}
\end{theo}
\begin{rem}
For brevity, the explicit definitions of $\phi_{\circ,\bullet}$'s are postponed to Definitions \ref{defi:phi_symbols}\ref{def:phi_MD_ML}-\ref{def:phi_MD_GPWM}, \ref{defi:phi_symbols}\ref{def:phi_ML}--\ref{defi:phi_symbols}\ref{def:phi_GPWM}, and functional limit results are provided all at once. Although each estimators' combination has its own peculiarities, these can be framed within a fairly general theory, which might be of interest per se. Due to the high degree of technicality, we present such theory in the appendix, herein focusing on ready-to-use estimators' examples.
\end{rem}

\begin{rem}\label{rem:theo}
In \cite{r15} and \cite{r17} modified versions of the estimators P, CFG and MD for $A_\alpha$ are proposed to guarantee that $\estpicka^{\scriptscriptstyle{\circ}}(\be_j)=1$ for all $n=1,2\ldots$ and $j=1,\dots,d$ where $\be_j=(0,\ldots,0,1,0,\ldots,0)$. 
The results in Theorem \ref{theo:conv_estimators} are also valid when such adjusted estimators are considered in place of \eqref{eq:pick_est_pick}, \eqref{eq:pick_est_cfg} and \eqref{eq:pick_est_md}, respectively, due to asymptotic arguments developed in the aforementioned works.
\end{rem}
\begin{rem}\label{rem:cond}
By the identity in \eqref{eq:ev_alpha_copula}, Proposition 1 in \cite{r15} guarantees, if the stable-tail dependence function $L$ satisfies Assumption 2 therein, that $C_{G_\alpha}$ satisfies Condition \ref{cond:cond_theo}.
%
%
\end{rem}

\subsection{Simulation study}\label{sec:simulation}
We show the finite sample performance of the composite-estimator $\estpickin^{\scriptscriptstyle{\circ,\bullet}}$ in \eqref{eq:est_inv_pick} through a simulation study.
Hereafter we consider for the P, CFG and MD estimators, the adjusted versions mentioned in Remark \ref{rem:theo}.

%
%

\begin{figure}[t!]
	\centering
	\includegraphics[width=0.21\textwidth, page=1]{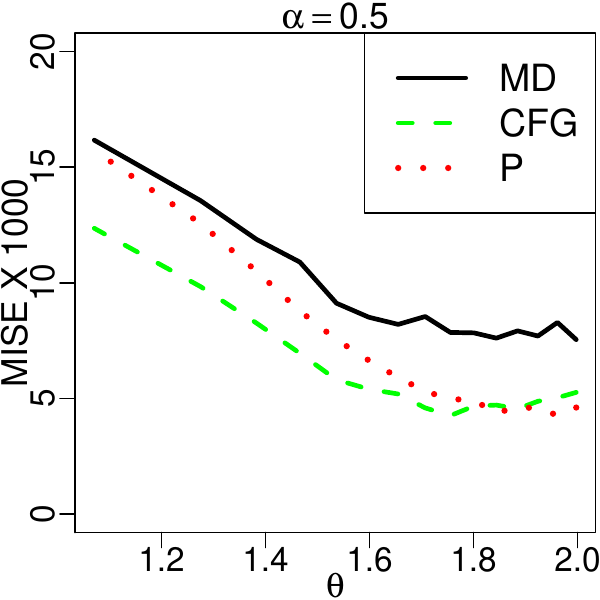}
	\includegraphics[width=0.21\textwidth, page=2]{doa_mise_gpwm_comp_n50_a.pdf}
	\includegraphics[width=0.21\textwidth, page=3]{doa_mise_gpwm_comp_n50_a.pdf}
	\includegraphics[width=0.21\textwidth, page=4]{doa_mise_gpwm_comp_n50_a.pdf}\\
	\includegraphics[width=0.21\textwidth, page=5]{doa_mise_gpwm_comp_n50_a.pdf}
	\includegraphics[width=0.21\textwidth, page=6]{doa_mise_gpwm_comp_n50_a.pdf}
	\includegraphics[width=0.21\textwidth, page=7]{doa_mise_gpwm_comp_n50_a.pdf}
	\includegraphics[width=0.21\textwidth, page=8]{doa_mise_gpwm_comp_n50_a.pdf}\\
	\includegraphics[width=0.21\textwidth, page=9]{doa_mise_gpwm_comp_n50_a.pdf}
	\includegraphics[width=0.21\textwidth, page=10]{doa_mise_gpwm_comp_n50_a.pdf}
	\includegraphics[width=0.21\textwidth, page=11]{doa_mise_gpwm_comp_n50_a.pdf}
	\includegraphics[width=0.21\textwidth, page=12]{doa_mise_gpwm_comp_n50_a.pdf}	
	\caption{MISE, ISB and IV  for $1000$ samples of
	size $50$ from an approximated distribution $Q$, obtained on the basis of  the standard Pareto
	distribution for $N$ and the bivariate Student-$t$ distribution 
	for $\bX$, for different values of the parameters  $\alpha$ and $\rho$, $\upsilon$. The parameter $\alpha$ is estimated with the GPWM estimator in \eqref{eq:gpwm}. 
	The parameter $\theta$ is the extremal coefficient  related to the corresponding extreme-value copula extremal-$t$.}
	\label{fig:mise_gpwm_extt_n50}
\end{figure}
%
%
%
%

Since it is not straightforward to simulate from the limit distribution $Q$, we study the performance of the composite-estimator $\estpickin^{\scriptscriptstyle{\circ,\bullet}}$ when it is used with data that are only 
approximately coming from the limiting distribution $Q$. Nevertheless, this is a more realistic scenario.

Specifically, we set $N=\ceil*{N'}$, where we assume that $N'$ follows a standard Pareto distribution 
with shape parameter $\alpha\in(0,1)$.
We simulate $N$ observations of a two-dimensional random vector $\bX$ with
a standard bivariate 
Student-$t$ distribution with a fixed value of the correlation $\rho$ and the degrees of freedom
$\upsilon$.
We recall that a Student-$t$ distribution is in the domain of attraction of a multivariate extreme-value 
distribution with an extreme-value copula that is the so-called Extremal-$t$ \citep[e.g., ][p. 189]{r7}. In the bivariate case, the extremal coefficient of the
Extremal-$t$ copula is $\theta=2T_{\upsilon+1}[\{(\upsilon+1)(1-\rho)/(1+\rho)\}^{1/2}]$, where $T_{\upsilon+1}$ is a univariate standard Student-$t$ distribution with $\upsilon+1$ degrees of freedom.
Next, with the simulated data we compute the observed value of the componentwise maxima $\bM_{N}$ in \eqref{eq:rnd_max}. 
We repeat these simulation steps $n'=500$ times generating $n'$ independent  observations from the pair $(N, \bM_{N})$ with which we
compute an observation from the random variable $\xi=\max(N_1,\ldots,N_{n'})$ and vector $\bbeta=\max(\bM_{N_1},\ldots,\bM_{N_{n'}})$,
where the later maximum is meant componetwise.
We repeat these simulation steps $n$ times, generating a data sample approximately drawn from the distribution 
$Q$, whose expression is given in the first line of \eqref{eq:joint_limiting} and where the expression of $G$ can been deduced from
\citet[][p. 189]{r7}.

Then, we estimate $\alpha$ using the observations generated from the sequence $\xi_1,\ldots, \xi_n$ with  the GPWM estimator $\estalpha^{\GPWM}$ 
in equation \eqref{eq:gpwm}, with $k=5$, and the ML estimator $\estalpha^{\ML}$ in \eqref{eq:mle}. Afterwards, we estimate the Pickands dependence function $\picka$ using the observations generated from the sequence $\bbeta_1,\ldots,\bbeta_n$ with the P estimator $\estpicka^\PICK$ in \eqref{eq:pick_est_pick}, CFG estimator $\estpicka^\CFG$ in \eqref{eq:pick_est_cfg} and MD estimator 
$\estpicka^\MD$ in \eqref{eq:pick_est_md}. 
Finally, we estimate $\pickainv$ using the composite-estimator $\estpickin^{\scriptscriptstyle{\circ,\bullet}}$ in equation \eqref{eq:est_inv_pick}. 

We repeat the simulation and estimation steps for different values of the model parameters $\alpha$, $\rho$ and $\upsilon$ and different sample sizes.
Precisely,  we consider $\alpha=0.5$, $0.633$, $0.767$, $0.9$ and, for the Student-$t$ distribution, we consider  the degrees of freedom $\upsilon=1$ and $15$ equally spaced values of the correlation $\rho$ in $[-0.99,0.99]$. With these parameters' values, the extremal coefficient $\theta$ (related to the Extremal-$t$ copula) takes values in $[1,2]$, where the lower and upper bounds represent the cases of complete dependence and independence. 
We also consider the sample sizes $n=50,100$. 
We repeat this experiment (the simulation and estimation considering different values of the parameters and the sample sizes) $1000$ times and we compute a Monte Carlo approximation of the Mean Integrated Squared Error (MISE), i.e.,
\begin{equation*}\label{eq:mise}
\begin{split}
\mbox{MISE}(\estpickin^{\scriptscriptstyle{\circ,\bullet}},\pickainv) &= \expect\left(  \int_{\simp} 
\left[\estpickin^{\scriptscriptstyle{\circ,\bullet}}(\bt)-\pickainv(\bt)\right]^2
\diff\bt\right )\\
&= \int_{\simp} \left[\expect\left(\estpickin^{\scriptscriptstyle{\circ,\bullet}}(\bt)\right)-\pickainv(\bt)\right]^2 \diff\bt+
\int_{\simp} \expect\left(\estpickin^{\scriptscriptstyle{\circ,\bullet}}(\bt)-\expect\left(\estpickin^{\scriptscriptstyle{\circ,\bullet}}(\bt)\right)\right)^2
\diff\bt,
\end{split}
\end{equation*}
where the first and second terms in the second line are known as Integrated Squared Bias (ISB)
and Integrated Variance (IV) \citep[][Ch. 6.3]{r21}.
%

%
%

\begin{figure}[t!]
	\centering
	\includegraphics[width=0.21\textwidth, page=1]{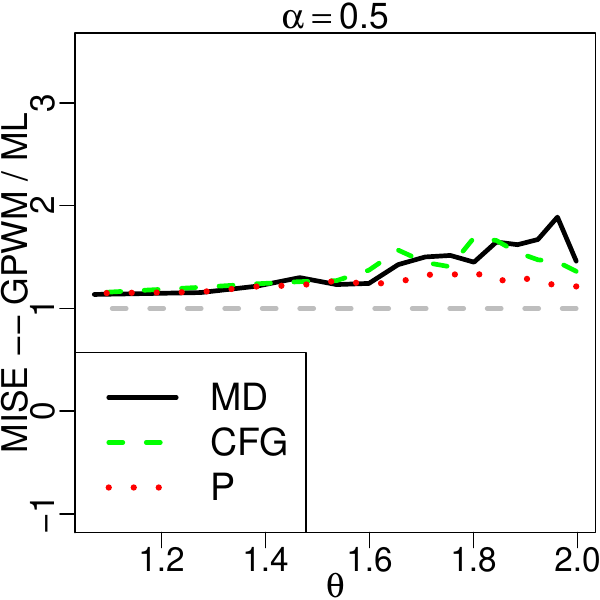}
	\includegraphics[width=0.21\textwidth, page=2]{mise_doa_gpwm_mle_comp_n50_a.pdf}
	\includegraphics[width=0.21\textwidth, page=3]{mise_doa_gpwm_mle_comp_n50_a.pdf}
	\includegraphics[width=0.21\textwidth, page=4]{mise_doa_gpwm_mle_comp_n50_a.pdf}\\
	\includegraphics[width=0.21\textwidth, page=5]{mise_doa_gpwm_mle_comp_n50_a.pdf}
	\includegraphics[width=0.21\textwidth, page=6]{mise_doa_gpwm_mle_comp_n50_a.pdf}
	\includegraphics[width=0.21\textwidth, page=7]{mise_doa_gpwm_mle_comp_n50_a.pdf}
	\includegraphics[width=0.21\textwidth, page=8]{mise_doa_gpwm_mle_comp_n50_a.pdf}\\
	\includegraphics[width=0.21\textwidth, page=9]{mise_doa_gpwm_mle_comp_n50_a.pdf}
	\includegraphics[width=0.21\textwidth, page=10]{mise_doa_gpwm_mle_comp_n50_a.pdf}
	\includegraphics[width=0.21\textwidth, page=11]{mise_doa_gpwm_mle_comp_n50_a.pdf}
	\includegraphics[width=0.21\textwidth, page=12]{mise_doa_gpwm_mle_comp_n50_a.pdf}
	\caption{Ratio between MISE, ISB and IV computed estimating the function $A^{\star}$ by
	the estimator $\estpickin^{\scriptscriptstyle{\GPWM,\bullet}}$ and $\estpickin^{\scriptscriptstyle{\ML,\bullet}}$ in formula \eqref{eq:est_inv_pick}.
	The same setting as Figure \ref{fig:mise_gpwm_extt_n50} is considered.}
	\label{fig:mise_comp_gpwm_ml_extt_n50}
\end{figure}
%
%
%

Figure \ref{fig:mise_gpwm_extt_n50} displays the results obtained with the GPWM-based estimators for the sample size $n=50$.
The MISE, ISB and IV ($\times 1000$) of the GPWM-based estimators are reported
from the first to the third row. The solid black, dashed green and dotted red lines report the results obtained estimating $A_\alpha$ with P, CFG and MD estimators, respectively.
The results for the different values of $\alpha$ are reported along the columns.

For each fixed value of $\alpha$ we see
that IV is close to zero at the strongest dependence level ($\theta=1$), then it increases with the decrease of the dependence level ($\theta$ increases approaching two).
On the contrary, ISB takes the largest value at $\theta=1$ and then it decreases with the decrease of the dependence level, for the cases $\alpha=0.5, 0.633$.
Overall, for the case $\alpha=0.5$, MISE takes the largest value at $\theta=1$ and then it decreases with the decrease of the dependence level. For the case $\alpha=0.633$, MISE does not change much over the whole range of dependence levels, since ISB and IV compensate each other.
While, for the cases $\alpha=0.767,0.9$, IV grows much more than ISB decreases, implying that MISE increases with the decrease of the dependence level.
The smallest values of ISB and IV are obtained with the CFG-based and P-based estimator, respectively.
Overall, on the basis of the MISE, the best performance is obtained with the CFG-based estimator, although there is little difference with the P-based estimator. As shown in the supplementary material, there is not much difference in the performance of the P-, CFG- and MD-based estimators 
already for the sample size $n=100$.

Although in this experiment we consider synthetic data 
that only approximately come from the distribution $Q$, the results summarised by ISB, IV and MISE 
highlight the robustness of our method to model misspecification (for only approximately max-stable data).
In particular, our composite-estimator displays a moderate distortion, despite that the uniform consistency guarantee of Theorem \ref{theo:conv_estimators} does not directly extend to the present setting.
Similar conclusions are obtained with sample size $n=100$ (available in the supplementary material).

A comparison between the estimation results obtained with the GPWM- and ML-based estimators is reported in Figure \ref{fig:mise_comp_gpwm_ml_extt_n50}.
Precisely, from the first to the third row, the ratio between the MISE, ISB and IV computed estimating the function $A^{\star}$ by the GPWM- and ML-based estimators are displayed.
For the case $\alpha=0.5$, on the basis of the ISB, 
the GPWM- and ML-based estimators perform very similarly, when
$1\leq \theta\leq1.5$, i.e. for strong up to moderate dependence levels. 
Instead, when $1.5\leq\theta\leq 2$, that is for moderate up to weak dependence levels, the ML-based estimators outperform the GPWM-based estimators. However, for very weak dependence levels ($\theta$ close to $2$) the GPWM-CFG estimator outperforms the ML-CFG estimator.
On the basis of the IV, 
the ML-based estimators outperform 
the GPWM-based estimators for all cases.
Concerning the configurations  with $\alpha>0.5$, the ML-based estimators considerably outperform GPWM-based estimators in terms of MISE. 
These conclusions are valid for all three P-, CFG- and MD-based estimators.
Specifically, IV is smaller for the ML-based estimators and their better performances are obtained for weaker extremal dependence structures (when $\theta$ approaches $2$). 
The ML-based estimators are much less biased than the GPWM-based estimators and the difference is much more pronounced for increasing values of $\alpha$ and weaker extremal dependence structures (when $\theta$ approaches $2$), although for the P- and CFG-based estimators such a difference diminishes when $\theta$ is close to $2$. 

The study was performed using the \textsf{R} \citep{r29} packages \texttt{Copula} \citep{r30} and \texttt{evd} \citep{r31}.

\section{Discussion}\label{sec:app}

Here we shortly discuss directions for future research, from both an applied and a theoretical viewpoint. First, we illustrate a potential real data problem with Internet traffic. In this domain, our theoretical framework can be applied to improve and extend the existing methods for the study of Internet traffic data and to perform an extreme-value statistical analysis. Next, we provide some concluding remarks, including probabilistic and methodological extensions of this work.\\

\noindent{\it \bf Massive Internet traffic data}. The analysis of Internet traffic data is crucial for improving the performance of large networks. 
Inferring internet congestion at different levels continues to receive increasing attention as new challenges are posed, for example, the booming demand for high-bandwidth contents (e.g., video streaming).  
Large scale collection and analysis of scientific data on Internet traffic is conducted by renowned research institutions, such as the
Center for Applied Internet Data Analysis (CAIDA, {\tt http://www.caida.org/home/}). CAIDA collects anonymized traffic traces from several monitors connected to commercial Internet backbones and large Internet service providers. A one-minute trace lists a huge amount of IP (Internet Protocol) packets, resulting in several gigabytes of compressed files. 
A first reduction of the data size is obtained resorting to flow records, i.e. measurements pertaining to coherent strings of packets, e.g., those stemming from the same traffic source, such as a single user \citep[see][Ch. 3]{r37}. Though, statistical analysis at this level is still computationally too burdensome.
Nonetheless, knowledge of the joint extremal behaviour of flow size (amount of data transmitted as part of a flow), say $X_1$,  and flow duration (time difference between the first and last packet of a flow), say $X_2$, would improve the existing techniques for large network optimization. 
 For instance, it would be of interest to forecast the through output rate $X_1/X_2$, corresponding to extremely large flow sizes
 \citep[see][Ch. 1.3.2, 1.3.5]{r42}. 

 To infer the extremal dependence between $X_1$ and $X_2$, it is first necessary to reduce the dimension of flow data by appropriately aggregating them. This is obtained by computing maxima of $X_1$ and $X_2$ over the random number of flows, $N$, occurring in a suitable time interval. The survival function can be reasonably expected to display a power-law behaviour, i.e. $\prob(N>n)=\slowf(n)\,n^{-\alpha}$ , with $\alpha$ smaller than one. This is consistent both with earlier theoretical fundings, e.g. \citet{r43}, and with the complexity of modern applications and backbone infrastructures. Different to the classical probabilistic description of Internet aggregate traffic
 \citep[e.g., ][]{r38, r39}, in which the number of active sources (flows) is deterministically sent to infinity, we account for randomness in the number of flows. Still, stochastic modelling via Pareto tails with $\alpha<1$ (hence $\expect{(N)}=+\infty$) appears a coherent refinement of the standard approach. 
 
 In this setting, Theorem 2.1 provides the mathematical basis for modeling the extremal dependence of aggregated-flow size and duration by means of the componentwise maxima approach. 
Furthermore, by the statistical inversion method in Section \ref{sec:nonparest} we can infer the extremal dependence between the extreme single-flow size and duration. Due to the extremely complex nature of flow data, we defer the actual data analysis to a future specialized work.\\  

\noindent
{\bf Concluding remarks.} The probabilistic and statistical modeling of aggregated data is an interesting and important topic.
The total and maximum amounts derived on a random number of observations described by the random vectors $\bS_N$ and $\bM_N$ in \eqref{eq:rnd_max}
are two simple examples of aggregated data.
There are not many results available on the extremal behaviour of $\bS_N$ and $\bM_N$, apart from those on $S_N$ in the univariate case.
This contribution makes a first step by establishing the multivariate extreme-value theory for $\bM_N$. 
Investigating the joint upper-tail behaviour of $\bM_N$, under the hidden regular variation framework \citep[e.g.,][]{r47, r44, r48}, would represent a first extension of this work.
Of particular interest would be investigating whether $F_N\in\MDA(\Phi_\alpha)$ with $\alpha\geq 1$ implies a stronger residual dependence than the case $F_N\in\MDA(\Lambda)$, as no difference emerges in terms of classical extremal dependence (see Theorem \ref{theo:new_doa}).
In this paper we consider for simplicity the same random number $N$ of independent copies for each components of $\bX=(X_1,\ldots,X_d)$. Of greater generality would be to consider $\bN=(N_1,\ldots,N_d)$, i.e. a different random number of independent copies for each component of $\bX$. In this context, it can be assumed that $\bN$ belongs to the domain of attraction of a multivariate extreme value distribution. This provides a second extension of our work.
A third extension could be the derivation of nonparametric estimators and their asymptotic results using threshold exceedances (for at least one component).
Finally, the extension of our results (probabilistic and inferential) to the case of $\bS_N$ represents a relevant open problem.
\appendix
\section{Proofs}\label{sec:appendix}
%

%
%
\subsection{Proof of Theorem \ref{theo:new_doa}}\label{sub:second_theo}
%

%
We start with some notation. Let $U_{X_j}(t):=F_{X_j}^{\leftarrow}(1-1/t)$, $t>1$, 
$D_j(x)=G_j^{\leftarrow}(e^{-1/x})$, $x>0$, $j=1,\ldots,d$, and
\begin{eqnarray*}
F_*(\cdot)&=&F_{\bX}(U_{X_1}(\cdot),\ldots,U_{X_d}(\cdot)),\quad \gs(\by)=G(D_1(y_1),\ldots,D_d(y_d)) ,\quad \by\in(0,\infty)^d.
\end{eqnarray*}
In extreme-value theory, it is common practice to derive the attractor $G$ of a distribution $F_X$ by analyzing separately the behaviour of its margins and its dependence structure. 
The latter is typically investigated by focusing on $F_*$, which obtains from $F_\bX$ by transforming the margins into unit-Pareto. In this way, we have that $F_* \in \mathcal{D}(\widetilde{G})$, where $\widetilde{G}$ has common unit-Fr\'{e}chet and the same extreme-value copula of $G$.
Although different types of common marginal distributions can be considered \citep[e.g., ][Ch. 4]{r6}, with Pareto margins we can exploit the theory of regularly varying tails.
Precisely, since $F_{\vk X} \in \mathcal{D}(G)$, we have that \citep[Propositions 5.10, 5.15 and 5.17]{r22}
\begin{equation}\label{eq:reg_vari}
\lim_{n\to \infty}\frac{1-F_*(n\by)}{1-F_*(n\bone)} = \frac{-\ln \gs(\by)}{\theta(G)}, \quad \by \in (0, \infty)^d,
\end{equation}
where $ -\ln \gs(\by)$ is homogeneous of order $-1$ and $\theta(G)=\theta(\widetilde{G})=-\ln\widetilde{G}(\bone)$.
We also define ${T}_*(y):=F_*(y\bone)$, $y>0$, a univariate non-decreasing function.  
Finally, we set $U_{*}(t):={T}_*^\leftarrow(1-1/t)$ and $U_N(t):={F}_N^\leftarrow(1-1/t)$, $t>1$.


The proof is organized in three parts: the derivation of the norming constants, two preliminary results and the conclusion. For the sake of brevity, some of the technical derivations are deferred to the supplementary material.

\subsubsection{Norming constants}\label{sub:second_theo_norming}
%


Let the norming sequences $\kappa_n$, $\varrho_n$ for $F_N$ be defined in the standard way \citep[e.g., ][pp. 48-54]{r22}. We recall that $F_{\bX}^n(\ba_n\bx+\bb_n)\to G(\bx)$ as $n\to\infty$ if and only if $n(1-F_{\bX}(\ba_n\bx+\bb_n))\to -\ln G(\bx)$ as $n\to\infty$, see e.g. \citet[][Ch.4]{r6} for details. Analogously, here we focus on
\begin{equation}\label{eq:ident_tail_laplace}
n(1-F_{\bM_N}(\bc_n \bx + \bd_n))=n(1-\LAP_N(-\ln p_n(\bx))),
\end{equation}
where $p_n(\bx)=F_\bX(\bc_n \bx+\bd_n)$. 
Accordingly, the derivation of the 
norming constants $\bc_n$ and $\bd_n$
requires an analysis of the behaviour of $1-\LAP_N(s)$, as $s\downarrow 0$.   
Observe that $\LAP_N(1/x)$, $x>0$,  is monotone nondecreasing. Set $U_{\LAP_N}(x):=(\LAP_N(1/\cdot))^{\leftarrow}(1-1/x)$, $x>0$. Define $z_n:=U_{\LAP_N}(n/b)$, where $b:=(\theta(G))^\alpha$ when $F_N \in \mathcal{D}(\Phi_\alpha)$, $\alpha \in (0,1)$, and $b:=(\theta(G))$ otherwise.
Then, setting $m_n:=U_{*}(z_n)$, as $n\to\infty$  we have $m_n\to\infty$ and 
$z_n \sim 1/\bar{T}_*(m_n)$, where $\bar{T}_*(y)=1-T_*(y)$, $y>0$. Consequently, we also have
\begin{equation}\label{eq:asym_equi}
1-\LAP_N(\bar{T}_*(m_n)) \sim  \frac{b}{n} , \quad n\to\infty.
\end{equation}

Next, we explictly provide some asymptotic approximations
which help to understand the results derived in the following subsections. When $F_N\in \MDA(\Phi_\alpha)$, $\alpha\in(0,\infty)$, then $\bar{F}_N(y)=1-F_N(y)$, $y>0$, satisfies 
\begin{equation}\label{eq:tail_alpha}
\bar{F}_N(y)  \sim \slowf(y)\,y^{-\alpha}, \quad y\to \infty,
\end{equation}
where $\slowf$ is slowly varying. In particular, 
for any $\alpha \in (0,1]$ we have that 
\begin{equation}\label{wahlJ2}
1-\LAP_N(s)\sim  \slowf ^*(1/s)s^\alpha ,\quad s\downarrow  0,
\end{equation}
for some slowly varying function $\slowf^*$, which is equal to $\Gamma(1-\alpha)\slowf$ if $\alpha \in (0,1)$, or satisfies $\lim_{x\to \IF} \slowf(x)/ \slowf^*(x)=0$, if $\alpha=1$. See Section 3.1 of the supplementary material for details.
When $F_N\in\MDA(\Phi_\alpha)$, $\alpha>1$, or $F_N\in\MDA(\Lambda)$,  
then $\EE{N} \in (0,\IF)$ 
and 
\begin{equation}\label{eq:lapl_finite}
1-\LAP_N(s)\sim \expect(N) s, \quad s \downarrow 0.
\end{equation}
Therefore, as $n \to \infty$ we have $z_{n}\sim \expect(N)n/b$ and
\begin{equation}\label{eq:m_n light tail}
\bar{T}_{*}(m_n)  \sim  1/z_n \sim \theta(G)/(n\expect(N)), \quad n\to\infty.
\end{equation} 
We finally derive $\bc_n$ and $\bd_n$. When
$F_N\in \MDA(\Phi_\alpha)$, $\alpha>0$,
\begin{itemize}
\item[(i)] if $F_{X_j}\in\MDA(\Phi_{\beta_j})$, then we set $d_{n,j}=0$ and
\begin{equation*}
c_{n,j}=
\begin{cases}
U_{X_j}(m_n), & \alpha \in (0,1]\\
U_{X_j}(n)\{\EE{N}\}^{1/\beta_j}, &  \alpha>1,
\end{cases}
\end{equation*}
\item[(ii)] if $F_{X_j}\in\MDA(\Lambda)$, then we set
\begin{eqnarray*}
c_{n,j}&=&
\begin{cases}
\omega_j(d_{n,j}),  & \alpha \in (0,1]\\
\omega_j\{\Upsilon_j^\leftarrow(1-1/\delta_jn)\}, & \alpha>1
\end{cases}
\\
d_{n,j}&=&
\begin{cases}
\Upsilon_j^\leftarrow(1-1/\delta_j m_n),  & \alpha \in (0,1]\\
c_{n,j}\ln \EE{N} + \Upsilon_j^\leftarrow(1-1/\delta_jn) , &  \alpha>1,
\end{cases}
\end{eqnarray*}
where $\Upsilon_j$ is the Von Mises function associated to $\bar{F}_{X_j}$, with $\bar{F}_{X_j}(x)=1-F_{X_j}(x)$ for $x\in\R$, $\omega_j$ its auxiliary function
\citep[e.g., ][pp. 40-43]{r22} and $\delta_j=\lim_{x\to \infty} \bar{F}_{X_j}(x)/\{1-\Upsilon_j(x)\}$;
\item[(iii)]  $F_{X_j}\in\MDA(\Psi_{\beta_j})$, then we set $d_{n,j}=x_{0,j}$, where $x_{0,j}=\sup\{x:F_{X_j}(x)<1 \}$, and 
\begin{equation*}
c_{n,j}=
\begin{cases}
\{U_{\tilde{X}_j}(m_n)\}^{-1}, & \alpha \in (0,1]\\
\{U_{\tilde{X}_j}(n) \EE{N}^{1/\beta_j}\}^{-1}, &  \alpha>1,
\end{cases}
\end{equation*}
where $\tilde{X}_j:=1/(x_{0,j}-X_j)$ and $U_{\tilde{X}_j}=F_{\tilde{X}_j}^\leftarrow(1-1/t)$, $t>0$.
\end{itemize}
\noindent
When $F_N\in \MDA(\Lambda)$, $\bc_n$ and $\bd_n$ are set equal to the sequences derived for the case  $F_N\in \MDA(\Phi_\alpha)$, with $\alpha>1$. With these norming constants, we obtain the following approximations as $n\to\infty$
\begin{equation}\label{eq:proxy_quantile}
U_{X_j}^{\leftarrow}(c_{n,j}x_j + d_{n,j})\sim m_n\,D_j^{\leftarrow}(x_j)=
\begin{cases}
m_n\,x_j^{\beta_j}, & F_{X_j}\in\MDA(\Phi_{\beta_j})\\
m_n\, e^{x_j}, & F_{X_j}\in\MDA(\Lambda)\\
m_n\,(-x_j)^{-\beta_j}, & F_{X_j}\in\MDA(\Psi_{\beta_j})
\end{cases}
\end{equation}
and
\begin{eqnarray}\label{eq:comple_proxy}
\nonumber 1-p_n(\bx)&\sim& 
1-F_{*}(m_n\,D_1^{\leftarrow}(x_1),\ldots,m_n\,D_d^{\leftarrow}(x_d))\\%
&\sim & \bar{T}_{*}(m_n)(-\ln G(\bx))/\theta(G)\notag \\
& \to & 0.
\end{eqnarray}

%
\subsubsection{Preliminary results}\label{sub:second_theo_preliminary}
Given the identity in \eqref{eq:ident_tail_laplace}, our first preliminary result 
provides the attractor of $F_{\bM_N}$.
%
\begin{lem}\label{lem:asy_exp}
If  $F_N \in \MDA(\Phi_\alpha), \alpha\in(0,1]$, then we have 
\begin{equation}\label{qen}
\limit{n}n (1-\LAP_N(-\ln p_n(\bx)))
=(-\ln G(\bx))^\alpha, \quad \bx\in\real^d.
\end{equation}
If $F_N \in \MDA(\Lambda)$ or $F_N \in \MDA(\Phi_\alpha), \alpha>1$, then \eqref{qen} holds with $\alpha=1$.
\end{lem}
\begin{proof}
If $F_N \in \MDA(\Phi_{\alpha})$, $\alpha\in(0,1]$, using sequentially \eqref{eq:comple_proxy}, \eqref{wahlJ2} together with the properties of slowly varying functions, and \eqref{eq:asym_equi}, as $n\to \IF$ we obtain
\begin{eqnarray} \label{eq:approxim}
n(1-\LAP_N(-\ln p_n(\bx)))
&\sim& n(1-\LAP_N(1- p_n(\bx)))\\
\nonumber
&\sim & n(1-\LAP_N(\bar{T}_{*}(m_n)(-\ln G(\bx))/\theta(G)))\\
\nonumber
&\sim & 
n(1-\LAP_N(\bar{T}_{*}(m_n)))(-\ln G(\bx))^\alpha/b\\
\nonumber&\sim & (-\ln G(\bx))^\alpha.
\end{eqnarray}
If $F_N \in \MDA(\Phi_{\alpha})$, $\alpha>1$, or $F_N \in \MDA(\Lambda)$, we have $\EE{N} < \IF$, thus from \eqref{eq:comple_proxy} and \eqref{eq:m_n light tail} it follows that, in this case, 
as $ n\to\infty$
\begin{eqnarray}\label{eq:second_proxy}
1-p_n(\bx) &\sim&-\ln G(\bx)/\{n\EE{N}\}.
%
\end{eqnarray} 
Consequently, in view of \eqref{eq:approxim} and the approximation in \eqref{eq:lapl_finite}, as $n \to \infty$ we have
\begin{eqnarray}\label{}
\nonumber n(1-\LAP_N(-\ln p_n(\bx)))&\sim& n\expect(N)(1-p_n(\bx))\\
&\sim & -\ln G(\bx),
\end{eqnarray}
which completes the proof.
\end{proof}
%

%


%

Next, we state an auxiliary result (see Section 3.2 of the supplementary material for the proof), that we use to establish our second preliminary result, characterizing the tail dependence between $\bM_N$ and $N$.

\begin{lem}\label{lem: tail ratio}
If $F_N\in \mathcal{D}(\Phi_\alpha)$ with $\alpha \in (0,1)$,  then we have (set $u_n(y):=\kappa_n y + \varrho_n$, $y>0$)
\begin{eqnarray}
\label{eq:lim1}
&&\lim_{n\to \infty} u_n(y) \bar{T}_{*}(m_n) = \lim_{n\to \infty}
\frac{U_N(n/(-\ln H(y)))}{U_{\LAP_N}(b/n)}= y\,\theta(G)\Gamma^{-1/\alpha}(1-\alpha),\\
\label{eq:lim2}
&&\lim_{n \to \infty}\frac{\bar{F}_N(-c/\ln p_n(\bx))}{\bar{F}_N(u_n(y))} = \lim_{n \to \infty} \frac{\bar{F}_N(c g(\bx) U_{\LAP_N}(n/b))}{\bar{F}_N(u_n(y))}=y^\alpha c^{-\alpha} \sigma^\alpha(\bx,\alpha), \quad \forall c>0
\end{eqnarray}
with $\sigma(\bx,\alpha)$ as in \eqref{eq:joint_limiting} and
$g(\bx)=\theta(G)/(-\ln G(\bx))$. If $F_N \in \mathcal{D}(\Phi_\alpha)$ with $\alpha \geq 1$ or $F_N \in \mathcal{D}(\Lambda)$, the limits in \eqref{eq:lim1}-\eqref{eq:lim2} (with $y >0$ or $y \in \mathbb{R}$, respectively) are equal to zero.
\end{lem}

\begin{lem}\label{lem: lim dep}
If $F_N\in \mathcal{D}(\Phi_\alpha)$ with $\alpha \in (0,1)$, then we have
$$
\lim_{n \to \infty} \prob
\left(
\bM_N \leq \bc_n \bx +\bd_n |N>u_n(y)
\right)
=\pi(\bx,y)-y^\alpha 
\left[ (-\ln G(\bx))^\alpha
- \sigma^\alpha(\bx, \alpha)\gamma(1-\alpha, y\sigma(\bx,\alpha))
\right] 
<1,
$$
where $\pi(\bx,y):=e^{-y\sigma(\bx,\alpha)}$. If $F_N \in \mathcal{D}(\Phi_\alpha)$ with $\alpha \geq 1$ or $F_N \in \mathcal{D}(\Lambda)$, then the above limit is equal to $1$.
\end{lem}
\begin{proof}
Few algebraic steps yield
\begin{equation}
	\prob
	\left(
	\bM_{N} \leq \bc_n \bx +\bd_n |N>u_n(y)
	\right)
	=p_n^{u_n(y)}(\bx)-\frac{1}{\bar{F}_N(u_n(y))}\int_{0}^{p_n^{u_n(y)}(\bx)}
	\bar{F}_N\left(\frac{\ln v}{\ln p_n(\bx)}\right)\diff v,
	\end{equation}
see Section 3.3 of the supplementary material for details.
Using \eqref{eq:comple_proxy} and Lemma \ref{lem: tail ratio}, as $n \to \infty$ we obtain
\begin{eqnarray*}
	p_n^{u_n(y)}(\bx) &\sim& \exp\{-u_n(y)(1-p_n(\bx))\}\\
	&\sim&\exp\{-u_n(y)\bar{T}_{*}(m_n) (-\ln G(\bx)/\theta(G)\}\\
	&\sim& \begin{cases}
		\pi(\bx, y), \hspace{2em} 
		F_N \in \mathcal{D}(\Phi_\alpha), \, \alpha \in (0,1)\\
		1, \hspace{4em} F_N \in \mathcal{D}(\Phi_\alpha), \, \alpha \geq 1 \text{ or } F_N \in \mathcal{D}(\Lambda).
\end{cases}.
\end{eqnarray*}

Hence, if $F_N\in \mathcal{D}(\Phi_\alpha)$ with $\alpha \in (0,1)$, 
by uniform convergence \citep[Proposition 0.5]{r22} and Lemma \ref{lem: tail ratio}, as $n \to \infty$ we also obtain
\begin{eqnarray*}
-\frac{1}{\bar{F}_N(u_n(y))}\int_{0}^{p_n^{u_n(y)}(\bx)}
\bar{F}_N\left(\frac{\ln v}{\ln p_n(\bx)}\right)\diff v &\sim&
-\frac{\bar{F}_N(-1/\ln p_n(\bx))}{\bar{F}_N(u_n(y))}\int_{0}^{\pi(\bx,y)}
(-\ln v)^{-\alpha}\diff v\\
&\sim& -y^\alpha \sigma^\alpha(\bx,\alpha) \left[ \Gamma(1-\alpha)-\gamma(1-\alpha,y\sigma(\bx,\alpha) \right]\\
&=&-y^\alpha 
\left[ (-\ln G(\bx))^\alpha
- \sigma^\alpha(\bx, \alpha)\gamma(1-\alpha,y\sigma(\bx,\alpha))
\right].
\end{eqnarray*}
While, if $F_N\in \mathcal{D}(\Phi_\alpha)$ with $\alpha \geq 1$ or $F_N\in \mathcal{D}(\Lambda)$, 
for any arbitrarily small $\epsilon>0$ and large enough $n$, we have $p_n^{u_n(y)}(\bx)\in [1-\epsilon,1]$ and $-\ln v/(-\ln p_n(\bx))>u_n(y)$, for all $v \in (0, 1-\epsilon)$;   therefore, an application of Lemma \ref{lem: tail ratio} yields that, as $n \to \infty$,
\begin{align*}
\bar{F}_N\left(\frac{\ln v}{\ln p_n(\bx)}\bigg{|} N> u_n(y)\right) =\frac{\bar{F}_N\left(\frac{-\ln v}{-\ln p_n(\bx)}\right)}{\bar{F}_N(u_n(y))} \to 0, \quad \forall v \in (0,1-\epsilon),
\end{align*}
where $\bar{F}_N(t|N>u_n(y)):=\prob(N>t|N>u_n(y))$, $t>u_n(y)$. Hence, by the dominated convergence theorem, as $n \to \infty$
\begin{eqnarray*}
	\frac{1}{\bar{F}_N(u_n(y))}\int_{0}^{p_n^{u_n(y)}(\bx)}
	\bar{F}_N\left(\frac{\ln v}{\ln p_n(\bx)}\right)\diff v 
	&=& \int_0^{p^{u_n(y)}_n(\bx)} \bar{F}_N\left(\frac{\ln v}{\ln p_n(\bx)}\bigg{|} N> u_n(y)\right) \diff t\\
	&\leq& \int_0^{1-\epsilon} \bar{F}_N\left(\frac{\ln v}{\ln p_n(\bx)}\bigg{|} N> u_n(y)\right) \diff t + \epsilon \\
	&\to&\epsilon,
\end{eqnarray*}
Since $\epsilon$ is aribitrarily small, the term on the left-hand side must converge to zero. The proof is now complete.
\end{proof}

%

%
\subsubsection{Conclusion}\label{sub:second_theo_main}

As $n\to\infty$, we have
\begin{eqnarray*}
-n\ln\prob(\bM_N\leq \bc_n\bx+\bd_n,N\leq \kappa_n y + \varrho_n)&\sim&n[1-\prob(\bM_N\leq \bc_n\bx+\bd_n,N\leq \kappa_n y + \varrho_n)]\\
&=&n\left[1-F_{\bM_N}(\bc_n \bx+\bd_n) 
\right]
+n\bar{F}_N(u_n(y)) \prob
\left(
\bM_{N} \leq \bc_n \bx +\bd_n |N>u_n(y)
\right).
\end{eqnarray*}
In view of \eqref{eq:ident_tail_laplace}, the limit of the first term on the right hand-side of the second line above is established in Lemma \ref{lem:asy_exp}. The second term is asymptotically equivalent to
$
-\log H(y)\prob
\left(
\bM_{N} \leq \bc_n \bx +\bd_n |N>u_n(y)
\right),
$
where the limit of the conditional probability is established in Lemma \ref{lem: lim dep}. Combining these results
we obtain the limiting expressions in \eqref{eq:joint_limiting} and \eqref{eq:joint_limiting_gumb}  and the proof is now complete.

%
%

%
\subsection{Proof of Corollary \ref{corr:loc_sc_mix}}\label{sub:proof_first_corr}
Let $\bc_n$, $\bd_n$, $\kappa_n$, $\varrho_n$ be the norming sequences defined in Section \ref{sub:second_theo} for the case $F_N\in \mathcal{D}(\Phi_\alpha)$, $\alpha\in(0,1)$. 
In particular $\varrho_n =0$, $n \bar{F}_N(\kappa_n) \sim 1$ and
$
\tilde{\kappa}_n \sim U_N(n\Gamma(1-\alpha))
$
as $n \to \infty$.   
The first result, i.e. $\tilde{c}_n^{-1}N_n^{\scriptsize{+}} \rightsquigarrow S$ as $n \to \infty$, now follows from Theorem 5.4.2 in \citet{r25}.

We recall that $p_n(\bx)=F_\bX(\bc_n\bx + \bd_n)$. In the proof of Lemma \ref{lem: lim dep} it has been established that
$$
p_n^{\kappa_n }(\bx) \sim \exp\left(- \frac{ -\ln G(\bx)}{\Gamma^{1/\alpha}(1-\alpha)} \right),\quad n \to \infty.
$$
Consequently, the second result now follows by noting that by the dominated convergence theorem
\begin{equation*}
\begin{split}
\lim_{n\to \infty}\mathbb{P}\left( \bM_{\floor{\tilde{c}_nS}} \leq \bc_n \bx+\bd_n\right)&=\lim_{n \to \infty} \int_0^\infty p_n ^{\floor{\tilde{c}_n s}}(\bx)\diff F_S(s)\\
&=\lim_{n \to \infty} \int_0^\infty \left(  p_n ^{\kappa_n}(\bx) \right)^{\floor{\tilde{\kappa}_n s}/\kappa_n} \diff F_S(s)\\
&=\int_0^\infty \left(e^{- \frac{(-\ln G(\bx))}{\Gamma^{1/\alpha}(1-\alpha)}} \right)^{s\Gamma^{1/\alpha}(1-\alpha)} \diff F_S(s)\\
&=
\LAP_S\left( -\ln G(\bx)\right)= G_\alpha(\bx).
\end{split}
\end{equation*}

%

%

%
\subsection{Proof of Corollary \ref{corr:extremalcoeff}}\label{sub:second_cor}

Let $Q_j$, $j=1, \ldots,d+1$, be the one-dimensional marginal distributions of the max-stable distribution $Q$. 
We focus on the case $F_N \in \mathcal{D}(\Phi_\alpha)$ with $\alpha\in(0,1)$ -- the other cases are trivial. From the first line of \eqref{eq:margins} we have $Q_j(x_j)=\exp\{-[-\ln G_j(x_j)]^\alpha\}$, $j=1,\ldots,d$, from which it follows that
$$
Q_j^\leftarrow(u_j)=G_j^\leftarrow(\exp\{-[-\ln u_j]^{1/\alpha}\}), \quad u_j \in (0,1).
$$
In particular,
$
Q_j^\leftarrow(e^{-1})=G_j^\leftarrow(e^{-1})$, $j=1, \ldots,d$.
By assumption, $Q_{d+1}(y)=\Phi_\alpha(y)$ and therefore 
$Q_{d+1}^\leftarrow(u_{d+1})=\Phi_\alpha^\leftarrow(u_{d+1})$, with $\Phi_\alpha^\leftarrow(u)=(-\ln u)^{1/\alpha}$. Hence $Q_{d+1}^\leftarrow(e^{-1})=1$. 

The extremal-coefficient is defined by
\begin{equation}\label{eq: coeff}
\theta(Q)=-\log Q( Q_1^\leftarrow(e^{-1}), \ldots, Q_d^\leftarrow(e^{-1}), Q_{d+1}^\leftarrow(e^{-1})).
\end{equation}
Plugging in the expressions of $Q_j^\leftarrow(e^{-1})$, $j=1, \ldots, d+1$, into the right-hand side of \eqref{eq: coeff}, the epxression of $\theta(Q)$ is then obtained.

As for the first result in \eqref{eq:excoefprop}, it immediately follows from the inequalities $\theta(G)\geq 1$ and $1 \leq  1-\mathcal{E}_{\Gamma^{1/\alpha}(1-\alpha)}(\theta(G))+ \mathcal{G}_{1-\alpha, \Gamma^{1/\alpha}(1-\alpha)}(\theta(G))$, for $\alpha \in (0,1)$.
To obtain the second one,
it is sufficient to note that, as $\alpha \to 1^-$,  $1-\mathcal{E}_{\Gamma^{1/\alpha}(1-\alpha)}(\theta(G))\to 1$ and
\begin{align*}
(\theta(G))^\alpha\mathcal{G}_{1-\alpha, \Gamma^{1/\alpha}(1-\alpha)}(\theta(G))&= \frac{(\theta(G))^\alpha}{\Gamma(1-\alpha)}\gamma\left(1-\alpha,
\frac{\theta(G)}{\Gamma^{1/\alpha}(1-\alpha)} \right)\\
&= \frac{(\theta(G))^\alpha}{\Gamma(1-\alpha)} \left[ \frac{\theta(G)}{\Gamma^{1/\alpha}(1-\alpha)}\right]^{1-\alpha} \sum_{k=0}^\infty\frac{\left[-\frac{\theta(G)}{\Gamma^{1/\alpha}(1-\alpha)}\right]^k}{k!(1-\alpha+k)}\\
&=\frac{\theta(G)}{\Gamma^{\frac{1}{\alpha}-1}(1-\alpha)}\left\lbrace \frac{1}{\Gamma(1-\alpha)}\frac{1}{1-\alpha} + \frac{1}{\Gamma(1-\alpha)}\sum_{k=1}^\infty\frac{\left[-\frac{\theta(G)}{\Gamma^{1/\alpha}(1-\alpha)}\right]^k}{k!(1-\alpha+k)} \right\rbrace\\
&=\frac{\theta(G)}{1+o(1)} \left\lbrace 
\frac{1}{1+o(1)} +o(1)
\right\rbrace
\end{align*}
In the above display, we use the following convergence results: as $\alpha \to 1^-$,
$\Gamma(2-\alpha)\to \Gamma(1)=1$, $\Gamma(1-\alpha)\to\infty$,
\begin{align*}
&\Gamma^{1/\alpha}(1-\alpha)= \exp\left\lbrace \frac{1}{\alpha} \ln \Gamma(1-\alpha) \right\rbrace\to \infty,\\
&\Gamma^{\frac{1}{\alpha}-1}(1-\alpha)= \exp\left\lbrace 
\left(\frac{1}{\alpha}-1 \right)\ln \Gamma(2-\alpha) -\frac{1}{\alpha}(1-\alpha)\ln(1-\alpha)
\right\rbrace \to 1.
\end{align*}

\subsection{Proof of Proposition \ref{prop:tail_behaviours}}\label{sub:prop_tail}
Let $U_\bX(\bt)=(U_{X_1}(t_1), \ldots, U_{X_d}(t_d))$, with $\bt=(t_1, \ldots, t_d)\geq 1$ and $U_{X_j}$, $j=1,\ldots,d$, as in Section \ref{sub:second_theo}, and set $F_*(\bt)=F_\bX(U_\bX(\bt))$.
Consider the case where $F_N\in\MDA(\Phi_\alpha)$ with $0<\alpha \leq  1$. 
By exploiting \eqref{eq:tail_alpha}-\eqref{wahlJ2} and arguments similar to those in the proof of Lemma \ref{lem:asy_exp},
we obtain that for $\by>\bzero$ and as $n\to\infty$
\begin{equation*}
\begin{split}
1-\prob(\bM_N\leq U_{\bX}(n\by))&=1-\LAP_N(-\ln F_*(n \by))\\
&\sim  \{1-F_*(n\by)\}^{\alpha}\slowf^{*}(\{1-F_*(n\by)\}^{-1})\\
&\sim  \{-\ln \gs(\by)/n\}^{\alpha}\slowf^{*}(n\{-\ln \gs(\by)\}^{-1})\\
&
\sim
\begin{cases}
\Gamma(1-\alpha)\prob(N>-n/\ln \gs(\by)), \quad \alpha \in (0,1)\\
(-\ln \gs(\by))(1-\LAP_N(1/n)), \hspace{2.35em} \alpha=1
\end{cases}.
\end{split}
\end{equation*}
Furthermore, $U_{\bM_N}(n\by)\sim U_{\bX}(U_{\LAP_N}(n)\by^{1/\alpha})$ as $n\to\infty$, with $U_{\LAP_N}$ as in Section \ref{sub:second_theo_norming} and, for $\bt=(t_1, \ldots, t_d) \geq 1$,
$$
U_{\bM_N}(\bt)=(U_{M_N^{(1)}}(t_1), \ldots, U_{M_N^{(d)}}(t_d)), \quad U_{M_N^{(j)}}(t_j)=F_{M_N^{(j)}}^{\leftarrow}(1-1/t_j), \quad j=1, \ldots,d.
$$ 
%
As a consequence $1-\prob(\bM_N\leq U_{\bM_N}(n\by))\sim n^{-1}\{-\ln \gs(\by^{1/\alpha})\}^{\alpha}$
as $n\to\infty$. 

Consider the case where $F_N\in\MDA(\Phi_\alpha)$ with $\alpha>  1$ or $F_N\in\MDA(\Lambda)$.
Again, with steps similar to those in Lemma \ref{lem:asy_exp} we have that for $\by>\bzero$ and
as $n\to\infty$
\begin{equation*}
\begin{split}
1-\prob(\bM_N\leq U(n\by))&\sim \EE N\{1-F_*(n\by)\}\\
&\sim \EE N(1-F_*(n\bone))\frac{-\ln \gs(\by)}{\theta(G)}\\
&\sim n^{-1}\EE N\{-\ln \gs(\by)\}.
\end{split}
\end{equation*}
In addition, $U_{\bM_N}(n\by)\sim U_{\bX}(\EE N n\by)$ as $n\to\infty$. Therefore, $1-\prob(\bM_N\leq U_{\bM_N}(n\by))\sim n^{-1}\{-\ln \gs(\by)\}$
as $n\to\infty$.

%
\subsection{Proof of Proposition \ref{prop:ppmax_max}}\label{sub:third_prop}
Observe that $\EE{Z_1^\alpha}= \Gamma(1- \alpha)$ for every $0<\alpha< 1$.
Setting  $k_\alpha=1/ \Gamma(1-\alpha )>0$ we have 
\BQN \notag  \frac{1}{k_\alpha}&=&
\EE{Z_1^\alpha} 
=  \int_0^\IF \pk{Z_1> t^{1/\alpha}} \diff t \\
&=&
 \int_0^\IF \Bigl(1 - e^{-  t^{-1/\alpha} } \Bigr)   \diff t=
 \int_0^\IF \Bigl(1 - e^{-  v^{1/\alpha} } \Bigr) v^{-2}  \diff v.
 \label{deso}
\EQN
For any positive $(y_1 \ldot y_d)$ we have 
\bqny{
 	- \ln \prob(R_1\leq y_1, \ldots, R_d \leq y_d ) 
 	&=& k_\alpha \EE{ \max_{1 \le j \le d} Z_j^\alpha /y_j^\alpha}\\
 	&=& k_\alpha \int_0^\IF  \pk{ \exists \, j\in \{1,\ldots,d\}: Z_j^\alpha > v y_j^\alpha}  \diff v \\
 	&=& \int_0^\IF \pk{ \exists \, j\in \{1,\ldots,d\}:  v\,k_\alpha  \,Z_j^\alpha > y_j^\alpha} v^{-2}  \diff v\\
 	&=& \int_0^\IF \pk{\exists \, j\in \{1,\ldots,d\}:  Z_j > y_j (k_\alpha v)^{-1/\alpha}} v^{-2}  \diff v\\
 	&=& \int_0^\IF \Bigl[1 - \pk{ \forall\, j\in \{1,\ldots,d\}:  Z_j \le y_j (k_\alpha v)^{-1/\alpha}}\Bigr] v^{-2} \diff v.
\EQNY
Now, we have 
$$ \pk{ \forall j\in \{1,\ldots,d\}:  Z_j \le y_j (k_\alpha v)^{-1/\alpha}} = 
\exp\left(- L( 1/y_1 \ldot 1/y_d)   k_\alpha ^{1/\alpha}  v^{1/\alpha}\right),
$$
where $L(\bz)$, with $\bz=1/\by$, is the stable-tail dependence function. 
Consequently,  by \eqref{deso} we obtain the final result	
 	\BQNY 
	- \ln \prob(R_1\leq y_1, \ldots, R_d \leq y_d ) 
 	&=& \int_0^\IF \Bigl[1 - e^{- L(1/ y_1, \ldots, 1/y_d)   k_\alpha ^{1/\alpha}  v^{1/\alpha} } \Bigr] v^{-2} dv\\
 	&=& \Bigl( L( 1/y_1 \ldot 1/y_d) \Bigr)^{\alpha} k_\alpha 
 	\int_0^\IF \Bigl[1 - e^{-  v^{1/\alpha} } \Bigr] v^{-2} dv\\ 
 	&=& L^\alpha(1/ y_1 \ldot 1/y_d)
 }
establishing the proof. 
%
%
\subsection{Proof of Proposition \ref{pro:doa_rnd:scale}}\label{sub:first_prop}
Let $\bM_n:=\max \{\bX_1, \ldots, \bX_n\}$. 
By the max-stability of $G$, there exist maps $\boldsymbol{\mathcal{A}}:(0, \infty)\mapsto (0, \infty)^d$ and $\boldsymbol{\mathcal{B}}:(0, \infty)\mapsto \mathbb{R}^d$, such that $G^s(\bx)=G(\boldsymbol{\mathcal{A}}(s)\bx+\boldsymbol{\mathcal{B}}(s))$, $s>0$. By equation (5.18) in \cite{r22}, for every $s>0$ we have
$$
\frac{\ba_{\floor{ns}}}{\ba_n}
\left(\frac{\bM_n - \bb_n}{\ba_n} - \frac{\bb_n- \bb_{\floor{ns}}}{\ba_{\floor{ns}}} \right)
\dot{=} \frac{1}{\boldsymbol{\mathcal{A}}(s)}
\left(\frac{\bM_n - \bb_n}{\ba_n} - \boldsymbol{\mathcal{B}}(s) \right)
\rightsquigarrow G(\boldsymbol{\mathcal{A}}(s)\cdot+\boldsymbol{\mathcal{B}}(s))=G^s, \quad n \to \infty.
$$
where $\dot{=}$ denotes asymptotic equivalence in distribution. Consequently, the final result follows from the equality $G_\alpha(\bx)= \expect({G^S(\bx)})$, $\bx \in \mathbb{R}^d$,  and the dominated convergence theorem
\begin{equation*}
\begin{split}
G_\alpha(\bx)&=\int_0^\infty G^s(\bx) \diff F_S(s)
=\lim_{n \to \infty}\int_0^\infty \prob \left(  
\frac{\ba_{\floor{ns}}}{\ba_n}
\left(\frac{\bM_n - \bb_n}{\ba_n} - \frac{\bb_n- \bb_{\floor{ns}}}{\ba_{\floor{ns}}} \right) \leq \bx
\right) \diff F_S(s)\\
&=\lim_{n \to \infty} \prob \left( 
\ba_n^{-1}\left[
\max\{\bw_n(\bX_1-\bv_n), \ldots, \bw_n(\bX_n-\bv_n) \} -\bb_n \right] \leq \bx
\right). 
\end{split}
\end{equation*}
%
%

%
\subsection{Proof of Proposition \ref{pro:pickands_rnd:scale}}\label{sub:first_prop}
We recall that for all $\bx\in\R^d$ we have $G(\bx)=C_G(G_1(x_1), \ldots,G_d(x_d))$, where the extreme-value copula is of the form $C_G(\bu)=\exp(-L((-\ln u_1),\ldots, (-\ln u_d))$, 
for $\bu\in(0,1]^d$, $L$ is the stable tail dependence function of $G$ and the corresponding Pickands dependence function satisfies \eqref{eq:pickA}.
We also recall that, for $\alpha\in(0,1)$, $G_\alpha(\bx)=\exp(-(-\ln G(\bx))^\alpha)$ is a max-stable distribution, 
with copula $C_{G_\alpha}$ given in \eqref{eq:ev_alpha_copula}. 

The copula $C_{G_\alpha}$ must be of extreme-value type, i.e. of the form
$C_{G_\alpha}(\bu)=\exp\{-L_\alpha(-\ln u_1, \ldots,-\ln u_d)\}$, $\bu \in [0,1]^d$,
for a stable tail-dependence function $L_\alpha$. We then deduce that
\begin{equation}\label{eq:thesables}
L_\alpha((-\ln u_1),\ldots,(-\ln u_d))
=L^\alpha((-\ln u_1)^{1/\alpha},\ldots,(-\ln u_d)^{1/\alpha}), 
\end{equation}
for all $\bu \in [0,1]^d$. By the homogeneity of the stable-tail dependence function, we have 
$$
L_\alpha(z_1 \ldot z_d)=(z_1+\cdots+z_d) A_\alpha(\bt),\quad \bz \in [0, \infty)^d,
$$
where $t_j = z_j / (z_1 + \cdots + z_d)$ for $j = 1, \ldots, d$ and 
$A_\alpha$ is the Pickands function of $G_\alpha$, $\alpha\in(0,1)$.
Combining \eqref{eq:thesables} with \eqref{eq:pickA}, we obtain
$$
L_\alpha(z_1 \ldot z_d) 
= \left(\sum_{j=1}^d z_j^{1/\alpha} \right)^\alpha  
	\left\{  A\left( \frac{\bz^{1/\alpha}}{\sum_{j=1}^{d} z_j^{1/\alpha}}\right)\right\}^\alpha,	
$$
where $A$ is the Pickands dependence function of $G$.
Therefore, choosing $\bz\in\simp$
we finally obtain 
$$ 
A_\alpha(\bt)= 
\left(\sum_{j=1}^{d} t_j^{1/\alpha} \right)^\alpha  
A^\alpha\left(
\frac{\bt^{1/\alpha}}{\sum_{j=1}^{d} t_j^{1/\alpha}} \right),
$$
which is the result in \eqref{eq:new_pick}.

%

%
\subsection{Proof of Theorem \ref{theo:conv_estimators}}\label{sub:third_theo}
\subsubsection{Notiation and general setting}\label{sub:third_theo_notation}
\noindent{\it \bf Empirical processes}. Recall that $(\bbeta_1,\xi_1),\ldots,(\bbeta_n,\xi_n)$, $n=1,2,\ldots$, are iid random vectors with
distribution $Q$ in \eqref{eq:joint_limiting} with fixed $\alpha\in(0,1)$.
For $i=1,2,\ldots$, and $j\in \{1,\ldots,d\}$, let
\begin{equation}\label{eq:unif_rvs} 
V_i:=\Phi_\alpha(\xi_i),\quad U_{i,j}:=G_{\alpha,j}(\eta_{i,j}),\quad 
 \widehat{U}_{i,j}:=G_{n,j}(\eta_{i,j}),
\end{equation}
where $\Phi_\alpha$ is the $\alpha$-Fr\'{e}chet distribution, $G_{\alpha,j}$ is the $j$-th margin of the distributions in the first line of \eqref{eq:margins} and $G_{n,j}$ is as in \eqref{eq:marg_empfun}. Set $\bU_i=(U_{i,1},\ldots,U_{i,d})$ and
$\widehat{\bU}_i=(\widehat{U}_{i,1},\ldots,\widehat{U}_{i,d})$. 
In the sequel, when the index $i$ is omitted, we refer to a single observation.
For every $\bu\in[0,1]^d$ and $v\in[0,1]$, define the random copula functions
\begin{equation}\label{eq:copula_fun}
C_{Q,n}(\bu,v):=\frac{1}{n}\sum_{i=1}^n\indic(\bU_i\leq \bu, V_i\leq v),\quad
C_{G_\alpha,n}(\bu)= C_{Q,n}(\bu,1),\quad B_{n}(v)= C_{Q,n}(\bone,v),
\end{equation}
where $\bone=(1,\ldots,1)$, and the copula processes
\begin{equation}\label{eq:emp_processes}
\cproc_{Q,n}(\bu,v)=\sqrt{n}(C_{Q,n}(\bu,v) - C_Q(\bu,v)), \quad
\cproc_{G_\alpha,n}(\bu)=\cproc_{Q,n}(\bu,1), \quad
\cbbridge_{n}=\cproc_{Q,n}(\bone,v).
\end{equation}
 Let $\cproc_{G_\alpha}(\bu):=\cproc_{Q}(\bu, 1) $, $\bu \in [0,1]^d$. The  covariance function of $\cproc_{G_\alpha}$ is 
 as in \eqref{eq:cov_gproc}, with $C_Q$ replaced by $C_{G_\alpha}$.
Furthermore, for every $\bu\in[0,1]^d$, define the empirical copula function and process
\begin{equation}\label{eq:emp_copula_fun_proc}
\widehat{C}_{G_\alpha,n}(\bu)=\frac{1}{n}\sum_{j=1}^n \indic(\widehat{\bU}_i \leq \bu), \quad
\widehat{\mathbb{C}}_{G_\alpha,n}=\sqrt{n}\left(\widehat{C}_{G_\alpha,n}-C_{G_\alpha}\right).
\end{equation}
\noindent In the sequel we view the above empirical processes as random signed measures, when appropriate \citep[e.g.,][Examples 1.7.4 and 1.10.6]{r27}.	
We then use the notation $Mf:=\int f \diff M$, for any  signed measure  $M$  on the measurable space $(\mathbb{X},\mathcal{X})$ and measurable function $f$.  Furthermore, for asymptotically measurable sequences $X_n,X_n'$ in $\ell^\infty(\mathbb{X})$, with $n \to \infty$, $X_n(\cdot)=X_n'(\cdot)+o_p(1)$ stands for $\sup_{x\in \mathbb{X}}|X_n(x)-X_n(x)|=o_p(1)$, and we recall that $X_n(\cdot)\rightsquigarrow X(\cdot)$ is shorthand for $\{X_n(x)\}_{x\in\mathbb{X}}\rightsquigarrow \{X(x)\}_{x\in\mathbb{X}}$ in $\ell^\infty(\mathbb{X})$, as already stated in Section \ref{sec:intro}.\\

\noindent{\it \bf Weighting and selection maps}. For all $f\in \ell([0,1]^{d+1})$, let $g_\epsilon:\ell([0,1]^{d+1}) \mapsto \ell([0,1]^{d+1})$ be the weighting-map given by
\begin{equation}\label{eq:weighted_map}
(g_\epsilon(f))(\bz)=
\begin{cases}
w^{-1}_{\epsilon}(\bz)f(\bz), \quad \bz\in (0,1]^{d+1}\setminus \{\bone\}\\
0, \hspace{5em} \text{otherwise},
\end{cases}
\end{equation} 
where, for any fixed $\epsilon\in[0,1/2)$, $w_{\epsilon}$ denotes the weighting-function
$$
w_\epsilon:[0,1]^{d+1} \mapsto [0,1]: \bz\mapsto \min_{1\leq j\leq d+1} z_j^\epsilon 
\left(1-\min_{1\leq j\leq d+1}z_j\right)^\epsilon.
$$
To keep the notation light,  when $f(\bz)$ is computed at $\bz=(\bu,1)$, we occasionally still write $g_\epsilon(f)$. The difference in meaning will be clear from the context.
For every $f\in \ell([0,1]^{d+1})$, let $\pi_{1,\ldots,d}: \ell([0,1]^{d+1})\mapsto \ell([0,1]^{d})$ and $\pi_{d+1}: \ell([0,1]^{d+1})\mapsto \ell([0,1])$ be the selection maps defined by
\begin{equation}\label{eq:selection}
(\pi_{1,\ldots,d}(f))(\bu):=f(u_1,\ldots,u_d,1),\quad (\pi_{d+1}(f)):=f(1,\ldots,1,v).
\end{equation}
Then, for every $\alpha\in(0,1)$, $\epsilon\in[0,1/2)$ and $\bu\in (0,1]^{d}\setminus \{\bone\}$, let $\omega_{\epsilon,\bu}: [0,1]^d \mapsto \R$ 
be the weighted-function defined by
\begin{equation}\label{eq:weighted_omega_function}
\omega_{\epsilon,\bu}(\bv):=\frac{\indic(\bv\leq \bu)-C_{G_\alpha} (\bu)}{(\pi_{1,\ldots,d}(w_{\epsilon}))(\bu)}.
\end{equation}
For every $\bu,\bv \in [0,1]^d$, set
\begin{equation}\label{eq:weighted_omega2_function}
\omega'_{\epsilon,\bu}(\bv)=
\begin{cases}
\omega_{\epsilon,\bu}(\bv), & \bu\in (0,1]^{d}\setminus \{\bone\}\\
0, & \text{otherwise}.
\end{cases}
\end{equation} 

\noindent{\it \bf Proof's overview}. By \citet{r26} and \citet{r15} we have that as $n \to \infty$
	\begin{equation}\label{eq_weighted_processes}
	g_\epsilon(\cproc_{Q,n}) \rightsquigarrow g_\epsilon(\cproc_{Q}),\quad 
	g_\epsilon(\cproc_{G_\alpha,n}) \rightsquigarrow 
	g_\epsilon(\cproc_{G_\alpha}),
	\end{equation}
	in $\ell([0,1])^{d+1}$ and $\ell([0,1])^{d}$, respectively. In particular, $g_\epsilon(\cproc_Q)$ is a zero-mean Gaussian process on $[0,1]^{d+1}$ with covariance function 
	$$
	\text{Cov}\{(g_\epsilon(\cproc_Q))(\bu),(g_\epsilon(\cproc_Q))(\bv)\}=\frac{C_Q(\bu \wedge \bv)-C_Q(\bu)C_Q(\bv)}{ w_\epsilon(\bu)w_\epsilon(\bv)}, \quad \bu, \bv \in [0,1]^{d+1}
	$$ 
	and $g_\epsilon(\cproc_{G_\alpha})$ is a zero-mean Gaussian process on the lower dimensional hypercube $[0,1]^{d}$ with covariance function defined analogously.
	The convergence results in \eqref{eq_weighted_processes} and, more generally, the asymptotic properties of the copula processes $\cproc_{Q,n}$ and $\cproc_{G_\alpha,n}$, are crucial for the derivation of the results presented in Section \ref{sub:third_theo}. As argued in Appendix \ref{sub:third_theo_body}, the P, CFG and MD estimators of $A_\alpha$ belong to a class of estimators $\estpicka$ satisfying the following conditions.

\begin{cond}\label{cond:cond1_prop}
The estimator $\estpicka$ of $\picka$ allows the subsequent representation:
\begin{inparaenum}
\item \label{subcond: cont}
For a continuous linear map $\phi:\ell^\infty([0,1]^d)\mapsto \ell^\infty(\simp)$ and some $\epsilon\in[0,1/2)$,
\begin{equation*}
\sqrt{n}\{\ln \estpicka(\cdot)-\ln \picka(\cdot)\}= (\phi\circ {g_{\epsilon}}(\cproc_{G_\alpha,n}))(\cdot)+o_p(1),
\end{equation*}  
where  $\cproc_{G_\alpha,n}$ is as in \eqref{eq:emp_processes};	
\item \label{subcond: integral}
For some $m\in \mathbb{N}_+$ and $-\infty \leq a<b\leq \infty$ the map $\phi$ is of the form  
\begin{equation}\label{eq:rep_phi_fun}
(\phi(f))(\by)=\sum_{i=0}^m \int_a^b f(\beta_{i,1}(z;t_1),\ldots,\beta_{i,d}(z;t_d)) K_i(z;\bt) \diff z.
\end{equation}
where, for $i=1,\ldots,m$,  $j\in\{1,\ldots,d\}$ and $\bt\in\simp$,  $z\mapsto \beta_{i,j}(z,t_j)$ is a bijective and continuous function, while $K_1, \ldots,K_m$ are functions satisfying
	$$
	\sup_{\bt\in\simp}\max_{0\leq i\leq m} |K_i(z;\bt)|\leq K(z), \quad z\in(a,b),
	\quad -\infty\leq a < b\leq \infty,
	$$
	for an integrable function $K$.
\end{inparaenum}
\end{cond}

An estimator $\estpicka$ allowing the above representations together with some suitable estimators $\estalpha$ of $\alpha$ (that meet some appropriate conditions) enables to deduce a general theory on the weak convergence for composite-estimators $\estpickin$ of $A^\star$, based on the copula process $\cproc_{Q,n}$. 
Such a theory is established in Appendix \ref{sub:third_theo_preliminary} and then applied to the specific cases of the GPWM and ML estimators of $\alpha$, and the P, CFG and MD estimators for $A_\alpha$, in Appendix \ref{sub:third_theo_body}. Some auxiliary results are deferred to Appendix \ref{sub:auxiliary_results}.
The arguments therein make extensive use of composition  maps  related to $\phi$. 
Hence, to improve the readability of the remaining part of Section \ref{sub:third_theo} we conclude this subsection by providing a comprehensive list of such composition maps.

\begin{defi}\label{defi:phi_symbols} Let $\phi$ be as in Condition \ref{cond:cond1_prop} and $\alpha \in (0,1)$ denote the true parameter value for the distribution $Q$ in \eqref{eq:joint_limiting}. Then:
\begin{inparaenum}
%
\item \label{def:limiting_GPWMbased}
For a measurable functional $\tau: \ell^\infty([0,1])\mapsto \mathbb{R}$,  the map
$\phi_{ \tau}: \ell^\infty([0,1]^{d+1})\mapsto \ell^\infty(\simp)$ is defined via
\begin{equation}\label{eq:GPWM_newmap}
(\phi_{\tau}(f))(\bt)=	  \alpha^{-1}(\phi \circ \pi_{1, \ldots,d}(f))(\bt)+	 K_\alpha(\bt) 
[\tau \circ \pi_{d+1}(w_\epsilon  f) ],
\end{equation}
for all $f \in \ell^\infty([0,1]^{d+1})$ and $\bt \in \simp$, with
\begin{equation}\label{eq:K_alpha}
K_\alpha(\bt)= \alpha^{-2} \left\{\norm{\bt}^{-1/\alpha}_{1/\alpha}\sum_{1\leq j\leq d:t_j>0} t_j^{1/\alpha}\ln t_j-\ln A_\alpha(\bt)\right\},
\end{equation}
$\pi_{d+1}$ as in \eqref{eq:selection} and $w_\epsilon$ as in \eqref{eq:weighted_map}. In particular, $\tau \circ \pi_{d+1}(w_\epsilon f)$ is a (real-valued) measurable functional;

%
%
\item \label{def:phi_comp_g} 
$\phi':\ell^{\infty}([0,1]^d)\mapsto \ell^{\infty}(\simp)$ is defined, for all $f \in \ell^\infty([0,1]^d)$ and $\bt \in \simp$, via
$$
(\phi'(f))(\bt)= \alpha^{-1}(\phi(f))(\bt)+K_{m+1}(\bt)f(1,\ldots,1),
$$
with $K_\alpha$ as in \eqref{eq:K_alpha} and
\begin{equation}\label{eq:K_function}
K_{m+1}(\bt)=K_\alpha(\bt)-\alpha^{-1}\sum_{i=0}^m\int_a^b K_i(z;\bt)\diff z.
\end{equation}
For a measurable function $\varphi:[0,1]\mapsto \mathbb{R}$, $\omega_{\epsilon,\bu}'$ is as in \eqref{eq:weighted_omega2_function} and $\delta(\bu)=1$, $\forall \bu\in[0,1]^d$, 
\begin{equation}\label{eq:gprime}
g_{\epsilon,\varphi}':  M \mapsto \left\{ 
\int_{[0,1]^{d+1}} \left[ \omega_{\epsilon,\bu}'(\ba)+\delta(\bu)\phi(b) \right] \diff M(\ba,b) 
\right\}_{\bu\in [0,1]^d}
\end{equation} 
maps a signed measure $M$ on $[0,1]^{d+1}$ to an element of $\ell^\infty([0,1]^d)$. With a little abuse of notation, we also apply this operator to the $C_Q$-Browninan bridge $\cproc_Q$ in Theorem \ref{theo:conv_estimators}, yielding $g_{\epsilon,\phi}'(\cproc_Q)$,
the zero-mean Gaussian process on $[0,1]^d$ with covariance function in \eqref{eq:cov_weight_proc};
%
%
\item \label{def:phi_MD} $\phi_{\MD}:\ell^{\infty}([0,1]^d)\mapsto \ell^{\infty}(\simp)$ is a special case of the map $\phi$, defined, for every $\bt\in\simp$, by
\begin{equation}\label{eq:phi_mado}
\begin{split}
(\phi_{\MD}(f))(\bt)=\frac{(1+\picka(\bt))^2}{\picka(\bt)}
\Bigg(&\sum_{j=1}^d\int_0^1\dot{C}_{G_\alpha;j}(v^{t_1},\ldots,v^{t_d})f(1,\ldots,1,
v^{t_j},1,\ldots,1)\diff v\\
&-\int_0^1f(v^{t_1},\ldots,v^{t_d})\diff v\Bigg);
\end{split}
\end{equation}
%
%
\item \label{def:phi_MD_ML} $\phi_{\MD,\ML}=\picka\phi_{\MD}' \circ g'_{0,\varphi_{\ML}}$,
where $\phi_{\MD}'$ is a special case of the map $\phi'$ in \ref{def:phi_comp_g} defined via
$$
(\phi_{\MD}'(f))(\bt)=\alpha^{-1}(\phi_{\MD}(f))(\bt) + K_{d+1}^{\MD}(\bt)f(1,\ldots,1)
$$
for every $f\in\ell^{\infty}([0,1]^d)$ and $\bt\in\simp$, with $\phi_{\MD}$ as in \eqref{eq:phi_mado},  $K_{d+1}^{\MD}$ given by 
$$
K_{d+1}^{\MD}(\bt)= K_\alpha(\bt)-\frac{\{1+A_\alpha(\bt)\}^2}{\alpha A_\alpha(\bt)}\left( \sum_{j=1}^d \int_0^1 \dot{C}_{G_\alpha;j}(v^{t_1},\ldots,v^{t_d})dv-1\right)
$$
and $K_\alpha$ is as in \eqref{eq:K_alpha}; $g'_{0,\varphi_{\ML}}$ is as in \ref{def:phi_comp_g}, for the particular choices of $\epsilon=0$ and $\varphi=\varphi_{ML}$ in \eqref{eq:phi_mle};
%
%
\item \label{def:phi_MD_GPWM} $\phi_{\MD,\GPWM}: \ell^\infty([0,1]^{d+1}) \mapsto \ell^\infty(\simp)$ is defined, for all $f \in \ell^\infty([0,1]^d)$ and $\bt \in \simp$, via
$$ 
(\phi_{\MD,\GPWM}(f))(\bt)=A_\alpha(\bt) (\phi_{ \tau_{GPWM}}^{\scriptscriptstyle{\MD}}\circ g_0(f))(\bt)=A_\alpha(\bt) (\phi_{ \tau_{GPWM}}^{\scriptscriptstyle{\MD}}(f))(\bt), 
$$
where $\phi_{\tau_{GPWM}}^{\scriptscriptstyle{\MD}}$ is a special case of the map $\phi_{\tau}$ in \eqref{eq:GPWM_newmap}, with $\phi=\phi_{\MD}$ in \eqref{eq:phi_mado} and $\tau=\tau_{\GPWM}$ in \eqref{eq:varphi_GPWM};
%
%
\item \label{def_phi_P_CFG} 
$\phi_{\scriptscriptstyle{\circ}}:\ell^{\infty}([0,1]^d)\to\ell^{\infty}(\simp)$ is a special case of the map $\phi$ defined by
\begin{equation}\label{eq:phi2}
\begin{split}
(\phi_{\scriptscriptstyle{\circ}}(f))(\bt)&=\int_0^\infty f(e^{-vt_1},\ldots,e^{-vt_d})\beta_\epsilon(e^{-v\max(\bt)})\diff v\\
&-\sum_{j=1}^d\int_0^\infty\dot{C}_{G_\alpha;j}(v^{t_1},\ldots,v^{t_d})f(1,\ldots,1,
v^{t_j},1,\ldots,1)\beta_\epsilon (e^{-vt_j})h_{\scriptscriptstyle{\circ}}(\bt;v)\diff v,
\end{split}
\end{equation}
with $\max(\bt)=\max(t_1, \ldots, t_d)$, $\beta_\epsilon(v)=v^\epsilon(1-v)^\epsilon$ for $v\in(0,1)$. Here $h_{\scriptscriptstyle{\circ}}$ is either
$h_{\PICK}(\bt;v)=-\picka^{-1}(\bt)$ or $h_{\CFG}(\bt;v)=1/v$ for $v>0$, $\bt \in \simp$;
%
%
\item \label{def:phi_ML} $\phi_{\scriptscriptstyle{\circ},\ML}=\picka\phi_{\scriptscriptstyle{\circ}}' \circ g'_{\epsilon,\varphi_{\ML}}$, 
where  $g'_{\epsilon,\varphi_{\ML}}$ is as in \ref{def:phi_comp_g}, with $\varphi=\varphi_{\ML}$ in \eqref{eq:phi_mle}, and $\phi'_{\scriptscriptstyle{\circ}}$ is a special case of the map $\phi'$ in \ref{def:phi_comp_g} defined via
$$
(\phi'_{\scriptscriptstyle{\circ}}(f))(\bt)=\alpha^{-1}(\phi_{\scriptscriptstyle{\circ}}(f))(\bt)+K^{\scriptscriptstyle{\circ}}_{d+1}(\bt)f(1, \ldots, 1),
$$
for every $f\in\ell^{\infty}([0,1])$ and $\bt\in\simp$,  with $\phi_{\scriptscriptstyle{\circ}}$ as in \eqref{eq:phi2} and
\begin{equation*}
\begin{split}
K_{d+1}^{\scriptscriptstyle{\circ}}(\bt)=K_\alpha(\bt)-
\int_0^\infty \beta_\epsilon\left(e^{-v\max(\bt)}\right) \frac{h_{\scriptscriptstyle{\circ}}(\bt;v)}{\alpha}\diff v +\sum_{j=1}^d\int_0^\infty\dot{C}_{G_\alpha;j}(v^{t_1},\ldots,v^{t_d})\beta_\epsilon \left(e^{-vt_j}\right)
\frac{h_{\scriptscriptstyle{\circ}}(\bt;v)}{\alpha}\diff v;
\end{split}
\end{equation*}
%
%
\item \label{def:phi_GPWM} $\phi_{\circ,\GPWM}:
\ell^\infty([0,1]^{d+1}) \mapsto \ell^\infty(\simp)$ is defined, for all $f \in \ell^\infty([0,1]^d)$ and $\bt \in \simp$, via
$$
(\phi_{\circ,\GPWM}(f))(\bt)=A_\alpha(\bt) (\phi_{\tau_{GPWM}}^{\scriptscriptstyle{\circ}}\circ g_\epsilon (f))(\bt)
$$
where $\phi_{ \tau_{GPWM}}^{\scriptscriptstyle{\circ}}$ is a special case of the map $\phi_{\tau}$ in \eqref{eq:GPWM_newmap}, with $\phi=\phi_{\circ}$ in \eqref{eq:phi2} and $\tau=\tau_{\GPWM}$ in \eqref{eq:varphi_GPWM}.
\end{inparaenum}
\end{defi}
%

%
%

%
\subsubsection{General preliminary results}\label{sub:third_theo_preliminary}

The results in this section rely on the following conditions on the estimators $\estalpha$ of $\alpha$.
\begin{cond}\label{cond:cond2_prop}
Let $\estalpha$ be an estimator of $\alpha$ satisfying one of the following properties:
\begin{inparaenum}
\item \label{en:first_cond2_prop} There is a continuous linear map $\tau:\ell^\infty([0,1])\mapsto \R$ such that 
$$
\sqrt{n}(\estalpha-\alpha)=\tau(\cbbridge_{n})+o_p(1),
$$
where $\cbbridge_{n}$ is as in \eqref{eq:emp_processes};
\item \label{en:second_cond2_prop} There is a measurable function $\zeta:(0,+\infty)\mapsto \R$ such that $\Phi_{\alpha}\zeta=0$, $\Phi_{\alpha}\zeta^2< \infty$ and
$$
\sqrt{n}(\estalpha-\alpha)=n^{-1/2}(\zeta(\xi_1)+\cdots+\zeta(\xi_n))+o_p(1).
$$
\end{inparaenum}
\end{cond}
When a general composite-estimator $\estpickin$ is obtained combining an estimator $\estpicka$ of $A_\alpha$ and $\estalpha$ of $\alpha$ satisfying Conditions \ref{cond:cond1_prop}\ref{subcond: cont} and
\ref{cond:cond2_prop}\ref{en:first_cond2_prop}, respectively, the functional weak convergence of $\sqrt{n}(\estpickin-\pickainv)$ can be established by fairly direct arguments. 
The GPWM-based composite-estimators are examples of estimators for which such asymptotic results hold, see Appendix \ref{sub:third_theo_body} and Lemma \ref{lem:third_step}.
When $\estalpha$ complies with \ref{cond:cond2_prop}\ref{en:second_cond2_prop},  to study the joint limit behaviour of $(\estpicka,\estalpha)$ and ultimately determining that of $\sqrt{n}(\estpickin-\pickainv)$ needs more complex asymptotic arguments.  The composite-estimators based on the M-estimator \citep[Ch. 5]{r2} and the MLE are examples of estimators for which the second type of asymptotic results apply. The following propositions derive the required theory.
Hereafter, the notations with the superscipt ``$\, ' \, $" are specific to the second type of asymptotic results.

%


%
\begin{sat}\label{pro:asy_norm_gpwm}
Let $\estpickin$ be the composite-estimator obtained by the composition of $\estpicka$ and $\estalpha$ satisfying Conditions \ref{cond:cond1_prop}\ref{subcond: cont} and \ref{cond:cond2_prop}\ref{en:first_cond2_prop}, respectively. Let $\phi_{\tau}$ be as in Definition \ref{defi:phi_symbols}\ref{def:limiting_GPWMbased} and $g_\epsilon$ be as in \eqref{eq:weighted_map}. Then, in $\ell^\infty(\simp)$
\begin{equation}\label{eq:asy_norm_gpwm}
\sqrt{n}\{\estpickin (\cdot) - \pickainv(\cdot)\}\rightsquigarrow  
	A_\alpha(\cdot)
	(\phi_\tau \circ {g_\epsilon}(\cproc_{Q}))(\cdot),
	\quad n\to\infty.
\end{equation}
\end{sat}
\begin{proof}
The claim in \eqref{eq:asy_norm_gpwm} relies on the following result.
%
%
%
\begin{lem}\label{lem:first_step}
With $\norm{\bt}_{1/h}$ being defined as in \eqref{eq:old_pick}, the map 
$$
g:(0,\infty)\mapsto \ell^{\infty}(\simp):h 
\mapsto\left(\ln \norm{\bt}_{1/h}^{1/h}\right)_{\bt\in\simp}
$$
is Hadamard differentiable at $\alpha$ with derivative
$$
\{(\Dot{g}_\alpha(h))(\bt)\}_{\bt\in\simp}=
\left\{-h \alpha^{-2} \norm{\bt}^{-1/\alpha}_{1/\alpha}\sum_{1\leq j\leq d: t_j>0} t_j^{1/\alpha}\ln t_j
\right\}_{\bt\in\simp}, \quad 0<h<\infty.
$$
\end{lem}
For the proof see Section 4 of the supplementary material.
For simplicity we focus on $\ln \pickainv$ and $\ln \estpickin$. For $\bt\in\simp$, we have that 
\begin{equation*}
\begin{split}
\sqrt{n}\{\ln \estpickin(\bt) -\ln \pickainv(\bt)\} &= \sqrt{n}\left\{ \estalpha^{-1} \ln \estpicka(\bt)- \alpha^{-1}\ln \picka(\bt) \right\}-\sqrt{n} \left( \ln \norm{\bt}_{1/\estalpha}^{1/\estalpha} -\ln \norm{\bt}_{1/\alpha}^{1/\alpha}\right)\\
&\equiv T_{1,n}(\bt)+T_{2,n}(\bt).
\end{split}
\end{equation*}
By Condition \ref{cond:cond2_prop}\ref{en:first_cond2_prop}, the functional version of Slutsky's lemma \citep[][p. 32]{r27} and the delta method \citep[][Theorem 3.1]{r2}  it follows that
$$
T_{1,n}(\cdot)= \alpha^{-1}\sqrt{n}\{\ln \estpicka(\cdot)-\ln \picka(\cdot)\}-\alpha^{-2}\ln \picka (\cdot)\sqrt{n}(\estalpha-\alpha)+o_p(1).
$$
By Lemma \ref{lem:first_step} and the functional delta method \citep[Ch. 20]{r2} it follows that
$$
T_{2,n}(\cdot)=\{K_\alpha(\cdot) + \alpha^{-2}\ln A_\alpha(\cdot)\}\sqrt{n}(\estalpha-\alpha)+o_p(1),
$$
where $K_\alpha$ is as in \eqref{eq:K_alpha}.
Then, by Conditions \ref{cond:cond1_prop}\ref{subcond: cont} and \ref{cond:cond2_prop}\ref{en:first_cond2_prop} and the identity $\tau \circ \pi_{d+1}(\cproc_{Q,n})=\tau \circ \pi_{d+1}(w_\epsilon g_\epsilon(\cproc_{Q,n}))$, with $\pi_{d+1}$ as in \eqref{eq:selection} and $w_\epsilon$ as in \eqref{eq:weighted_map}, we obtain
\begin{equation}\label{eq:T_1+T_2}
\begin{split}
T_{1,n}(\cdot)+T_{2,n}(\cdot)&=\alpha^{-1}(\phi\circ{g_\epsilon}(\cproc_{G_\alpha,n}))(\cdot)+K_\alpha(\cdot) \sqrt{n}(\estalpha-\alpha)+o_p(1)\\
&=\alpha^{-1}(\phi \circ {g_\epsilon}(\cproc_{G_\alpha,n}))(\cdot)+K_\alpha(\cdot) \tau(\cbbridge_{n})+o_p(1) \\
&=(\phi_\tau \circ g_\epsilon(\cproc_{Q,n}))(\cdot)+o_p(1),
\end{split}
\end{equation}
Thus, the result in \eqref{eq:asy_norm_gpwm} follows from \eqref{eq_weighted_processes}, by applying the continuous mapping theorem and the functional delta method in the last line of \eqref{eq:T_1+T_2}.
\end{proof}
%
%
%
%

%
\begin{sat}\label{pro:asy_norm_mle}
Let $\estpickin$ be the composite-estimator obtained by the composition of $\estpicka$ and $\estalpha$ complying with Condition \ref{cond:cond1_prop}\ref{subcond: cont}-\ref{subcond: integral} and
Condition \ref{cond:cond2_prop}\ref{en:second_cond2_prop}, respectively, 
with
$\varphi=\zeta\circ \Phi_\alpha^{\leftarrow}$ satisfying
\begin{equation}\label{eq:mean_con}
-\infty < \EE{\omega'_{\epsilon,\bu}(\bU)\varphi(V)} < \infty, \quad  \forall \bu \in (0,1]^d \setminus \{\bone\},
\end{equation}
where $\omega_{\epsilon,\bu}'$ is given in \eqref{eq:weighted_omega2_function}.  Let $\phi'$ and $g'_{\epsilon,\varphi}$ be as in Definition \ref{defi:phi_symbols}\ref{def:phi_comp_g}.
Then, in $\ell^{\infty}(\simp)$ as $n\to\infty$ 
$$
\sqrt{n}\left\{\estpickin (\cdot) - \pickainv (\cdot)\right\}\rightsquigarrow \picka(\cdot)(\phi' \circ {g'_{\epsilon,\varphi}}(\cproc_{Q}))(\cdot),
$$
where
$g'_{\epsilon,\varphi}(\cproc_{Q})$ is a zero-mean Gaussian process with covariance function defined in \eqref{eq:cov_weight_proc}.
\end{sat}
\begin{proof}
For any $\bu_i \in [0,1]^d, i\le k$, $k \in \mathbb{N}_+$,  the random vectors 
$(\omega'_{\epsilon,\bu_1}(\bU_i),\ldots,\omega'_{\epsilon,\bu_k}(\bU_i),\varphi(V_i))$, 
$i=1,\ldots,n$, are iid with zero-mean and finite pairwise covariances, 
by arguments in \citet{r26}, \citet{r15}, Condition \ref{cond:cond2_prop}\ref{en:second_cond2_prop} and
\eqref{eq:mean_con}.
Let
$$
(g'_\epsilon(\cproc_{Q,n}))(\bu)= \frac 1{\sqrt{n}} \sum_{i=1}^n \omega'_{\epsilon,\bu}(\bU_i),\quad
\bar{\varphi}_n(\bu)= \delta(\bu)\frac{1}{\sqrt{n}}\sum_{i=1}^n\varphi(V_i), \quad \bu\in[0,1]^d,
$$
where $\delta(\bu)=1$ for all $\bu\in[0,1]^d$. Note that $(g'_\epsilon(\cproc_{Q,n}))(\bu)=(g_\epsilon(\cproc_{Q,n}))(\bu)$, with
$\bu\in[0,1]^d$.
Then, $g_\epsilon'(\cproc_{Q,n})$ and $\bar{\varphi}_n$ are asymptotically tight \citep[Definition 1.3.7]{r27}
and, by the  central limit theorem, 
$$
((g'_\epsilon(\cproc_{Q,n}))(\bu_1),\ldots, (g'_\epsilon(\cproc_{Q,n}))(\bu_k), \, \bar{\varphi}_n(\bu_1),\ldots,\bar{\varphi}_n(\bu_k))\rightsquigarrow N(\bzero,\Sigma)
$$
as $n\to\infty$. 
By arguments in \citet[p. 42 point 3]{r27}, these facts are sufficient to claim that the class of functions 
%
$\preal_{\epsilon, \varphi}:=\{(\ba, b)\mapsto \omega_{\epsilon,\bu}'(\ba)+\delta(\bu) \varphi(b): \bu \in [0,1]^d \}$
%
is $C_Q$-Donsker \citep[pp. 80-82]{r27}. Indeed, 
introducing the map 
$$
g_{\epsilon,\varphi}':  M \mapsto \{ M f: f\in \preal_{\epsilon, \varphi}\}
$$ 
defined on the space of signed measure $M$ on $[0,1]^{d+1}$,  
we have that
$
(g'_{\epsilon,\varphi}(\cproc_{Q,n}))(\omega_{\epsilon,\bu}'+\delta(\bu) \varphi)= (g'_\epsilon(\cproc_{Q,n}))(\bu) + \bar{\varphi}_n(\bu),
$
$\forall\,\bu\in[0,1]^d$.
Then, as $n\to\infty$, 
$g'_{\epsilon,\varphi}(\cproc_{Q,n})\rightsquigarrow g'_{\epsilon,\varphi}(\cproc_{Q})$ in $\ell^{\infty}(\preal_{\epsilon, \varphi})$, 
where $g'_{\epsilon,\varphi}(\cproc_{Q})$ is a zero-mean Gaussian process with covariance function
\begin{equation}\label{eq:cov_weight_proc}
\begin{split}
\Cov&\left\lbrace (g'_{\epsilon,\varphi}(\cproc_{Q}))(\omega_{\epsilon,\bu}'+\delta(\bu) \varphi),(g'_{\epsilon,\varphi}(\cproc_{Q}))(\omega_{\epsilon,\bv}'+\delta(\bv) \varphi)\right\rbrace\\
&=\begin{cases}
\expect(\{\omega_{\epsilon,\bu}(\bU)+\varphi(V)\}\{\omega_{\epsilon,\bv}(\bU)+\varphi(V)\}), & \bu, \bv \in \setV\\
\expect(\{\omega_{\epsilon,\bu}(\bU)+\varphi(V)\}\varphi(V)), & \bu\in \setV, \bv \in \setV^c\\
\expect(\varphi^2(V)), & \bu,\bv \in \setV^c\\
\end{cases},
\end{split}
\end{equation}
where $\setV=(0,1]^d\backslash\{\bone\}$. Since each element of $\preal_{\epsilon, \varphi}$ corresponds to a unique $\bu \in [0,1]^d$, we can equivalently think of the codomain of $g_{\epsilon,\varphi}'$ as $\ell^\infty([0,1]^d)$ (as done in \eqref{eq:gprime}) and consider the processes $g'_{\epsilon,\varphi}(\cproc_{Q,n})$, $g'_{\epsilon,\varphi}(\cproc_{Q})$ as indexed on $[0,1]^d$.
From \eqref{eq:rep_phi_fun}, Condition \ref{cond:cond2_prop}\ref{en:second_cond2_prop} and the first line of \eqref{eq:T_1+T_2}, it follows that 
$$
\sqrt{n}\{\ln\estpickin(\cdot) - \ln \pickainv(\cdot)\}= (\phi' \circ g'_{\epsilon,\varphi}(\cproc_{Q,n}))(\cdot)+o_p(1).
$$
The final results are obtained by
applying the continuous mapping theorem and the functional delta method to the above expression.
\end{proof}
%
%
%

%
%

%
%

%
\subsubsection{Main body of the proof}\label{sub:third_theo_body}
In what follows, the asymptotic results concerning the ML- and GPWM-based composite-estimators are established by verifying the assumptions of Proposition \ref{pro:asy_norm_mle}
and Proposition \ref{pro:asy_norm_gpwm}, respectively. Firstly, we address the madogram-based case; then, we simultaneosly examine P- and CFG-based estimators, due to their similar traits.

We start analyzing the case when $\alpha$ is estimated with the ML estimator in \eqref{eq:mle}
and $\picka$ with the MD estimator in \eqref{eq:pick_est_md}. 
We recall that the estimator in \eqref{eq:mle} is the unique solution of log-likelihood equation
$n^{-1}\sum_{i=1}^n \dot{\mathcal{L}}_{\tilde{\alpha}}(\xi_i)=0$, where
$$
\dot{\mathcal{L}}_{\tilde{\alpha}}(x)=\partial/\partial \tilde{\alpha} \ln\dot{\Phi}_{\tilde{\alpha}}(x)=1/\tilde{\alpha}+\ln x (x^{-\tilde{\alpha}}-1),
\quad x>0.
$$
Noting that $\xi^{-1}$ is a Weibull random variable, then by using similar arguments to \citet[Theorem 5.41 and 5.42]{r2} it follows that $\estalpha^{\ML}\conP \alpha$ as $n\to\infty$ and
\begin{equation}\label{eq:mle_expantion}
%
\sqrt{n}\left(\estalpha^{\ML}-\alpha\right)=i_\alpha^{-1}n^{-1/2}\{\dot{\mathcal{L}}_\alpha(\xi_1)+\cdots+\dot{\mathcal{L}}_\alpha(\xi_n)\}+o_p(1),
\end{equation}
where $i_\alpha=\alpha^{-2}\{(1-\varsigma^2)+\pi^2/6\}$ is the Fisher information and $\varsigma$ is Euler's constant.

Assuming that Condition \ref{cond:cond_theo}\ref{en:first_cond} holds true, 
then, by 
Lemma \ref{lem:mado_est_prob_rep},  $\estpicka^\MD$ satisfies Condition \ref{cond:cond1_prop}\ref{subcond: cont}-\ref{subcond: integral} with $\phi=\phi_{\MD}$ in \eqref{eq:phi_mado}, and $\epsilon=0$.
Furthermore, $\estalpha^{\ML}$ satisfies Condition \ref{cond:cond2_prop}\ref{en:second_cond2_prop} by 
\eqref{eq:mle_expantion} with $\zeta=\zeta_{\ML}=i^{-1}_\alpha\dot{\mathcal{L}}_\alpha$.
Define
\begin{equation}\label{eq:phi_mle}
\varphi_{\ML}(v):=\zeta_{\ML}\circ\Phi^{\leftarrow}_\alpha
(v)=
i^{-1}_\alpha\alpha^{-1}\{1+(1+\ln v)\ln(-\ln v)\},\quad v\in(0,1),
\end{equation}
then \eqref{eq:mean_con} is satisfied with $\varphi=\varphi_{\ML}$, by Lemma \ref{lem:covariance}. 
Therefore, from Proposition \ref{pro:asy_norm_mle} it follows that, in $\ell^{\infty}(\simp)$,
$$
\sqrt{n}\left\{\estpickin^{\scriptscriptstyle{\MD,\ML}}(\cdot)-\pickainv(\cdot)\right\}\rightsquigarrow
(\phi_{\MD,\ML}(\cproc_Q))(\cdot),\quad n\to\infty,
$$
where $\phi_{\MD,\ML}=\picka\phi_{\MD}' \circ g'_{0,\varphi_{\ML}}$, 
with details given in Definition \ref{defi:phi_symbols}\ref{def:phi_MD_ML}, and where
$g'_{0,\varphi_{\ML}}(\cproc_Q)$ is a zero-mean Gaussian process with covariance function
$$
\Cov\{g'_{0,\varphi_{\ML}}(\cproc_Q)(\bu), g'_{0,\varphi_{\ML}}(\cproc_Q)(\bv)\}=
C_{G_\alpha}(\min(\bu,\bv))-C_{G_\alpha}(\bu)C_{G_\alpha}(\bv) + T_\alpha(\bu)+T_\alpha(\bv)+1,
$$
for every $\bu,\bv\in[0,1]^d$, with 
$$
T_{\alpha}(\cdot)=\frac{1}{i_\alpha \alpha}\left(C_{G_\alpha}(\bu) - 
\int_0^1 \frac{\partial}{\partial v} C_{Q}(\bu,v)(1+\ln v)\ln(-\ln v) \diff v\right).
$$
Finally, from the weak convergence result and the functional version of Slutsky's lemma it follows
that $\|\estpickin^{\scriptscriptstyle{\MD,\ML}}-\pickainv\|_\infty\conP 0$ as $n\to\infty.$

Next, we study the case when $\alpha$ is estimated with the GPWM estimator in \eqref{eq:gpwm}
and $\picka$ with the MD estimator in \eqref{eq:pick_est_md}. Here, we additionally assume that
$\alpha>1/(k-1)$. By Lemma \ref{lem:third_step}, the estimator $\estalpha^{\GPWM}$ satisfies Condition 
\ref{cond:cond2_prop}\ref{en:first_cond2_prop} with $\tau=\tau_{\GPWM}$ given in \eqref{eq:varphi_GPWM}.
Then, by Proposition \ref{pro:asy_norm_gpwm} it follows that, in $\ell^{\infty}(\simp)$,
$$
\sqrt{n}\left\{\estpickin^{\scriptscriptstyle{\MD,\GPWM}}(\cdot)-\pickainv(\cdot)\right\}\rightsquigarrow
(\phi_{\MD,\GPWM}(\cproc_Q))(\cdot),\quad n\to\infty,
$$
where $\phi_{\MD,\GPWM}(f)$, $f \in \ell^{\infty}([0,1]^{d+1})$,
is given in Definition \ref{defi:phi_symbols}\ref{def:phi_MD_GPWM}.
Furthermore, given the result in Lemma \ref{lem:third_step} and since 
$\|n^{-1/2}\cbbridge_{n}\|_\infty\conAS0$, then we have that $\estalpha^\GPWM\conAS \alpha$ as $n\to\infty$. Consequently,
by Lemma \ref{lem:mado_est_AS}, 
$$
\left\|\estpickin^{\scriptscriptstyle{\MD,\GPWM}}-\pickainv\right\|_{\infty}\conAS 0,\quad n\to\infty.
$$

Now, we study the case when $\alpha$ is estimated with the ML estimator in \eqref{eq:mle}
and $\picka$ with the P and CFG estimator in \eqref{eq:pick_est_pick} and \eqref{eq:pick_est_cfg}, respectively.
Assuming that Conditions \ref{cond:cond_theo}\ref{en:first_cond} and 
\ref{cond:cond_theo}\ref{en:second_cond} hold true, then by \citet[Proposition 3.1 and 4.2]{r28}
we have
$$
\widehat{\cproc}_{G_\alpha,n}(\bu)= \cproc_{G_\alpha,n} (\bu)-\sum_{j=1}^d\dot{C}_{G_\alpha;j}\cproc_{G_\alpha,n}
(1,\ldots,1,u_j,1,\ldots,1)+R_n(\bu),
$$
where $\widehat{\cproc}_{G_\alpha,n}$ is as in \eqref{eq:emp_copula_fun_proc} and almost surely 
$
\Vert R_n \Vert_\infty=O(n^{-1/4}(\ln n)^{1/2}(\ln\ln n)^{1/4}),
$
$n\to\infty$.
Then, using similar arguments to those in \citet[pp. 3082--3083]{r15} (and the functional delta method for the P estimator) we have that
$\estpicka^{\PICK}$ and $\estpicka^{\CFG}$ satisfy Condition \ref{cond:cond1_prop}\ref{subcond: cont}-\ref{subcond: integral} with
$\phi=\phi_{\PICK}$ and $\phi=\phi_{\CFG}$.
Precisely, for any fixed $\epsilon\in (0,1/2)$ 
$$
\sqrt{n}\{\ln \estpicka^{\scriptscriptstyle{\circ}}(\cdot)-\ln \picka(\cdot)\}=(\phi_{\scriptscriptstyle{\circ}}\circ g_\epsilon(\cproc_{G_\alpha,n}))(\cdot)+o_p(1),
$$
where $\phi_{\scriptscriptstyle{\circ}}$ is given in \ref{defi:phi_symbols}\ref{def_phi_P_CFG} and satisfies the representation in \eqref{eq:rep_phi_fun}. By Lemma \ref{lem:covariance}, the expectation in \eqref{eq:mean_con} is  finite for $\varphi=\varphi_{\ML}$ and any given $\epsilon \in (0,1/2)$.
Then, by Proposition \ref{pro:asy_norm_mle}, 
$$
\sqrt{n}\left\{\estpickin^{\scriptscriptstyle{\circ,\ML}}(\cdot)-\pickainv(\cdot)\right\}\rightsquigarrow
(\phi_{\circ,\ML}(\cproc_Q))(\cdot),\quad n\to\infty
$$
in $\ell^{\infty}(\simp)$, where $\phi_{\scriptscriptstyle{\circ},\ML}=\picka\phi_{\scriptscriptstyle{\circ}}' \circ g'_{\epsilon,\varphi_{\ML}}$, 
with details given in Definition \ref{defi:phi_symbols}\ref{def:phi_ML}, and where
$g'_{\epsilon,\varphi_{\ML}}(\cproc_Q)$ is a zero-mean Gaussian process with covariance function as in \eqref{eq:cov_weight_proc}, for $\varphi=\varphi_{\ML}$.
Finally, from this result and the functional version of Slutsky's lemma,
it follows that $\|\estpickin^{\scriptscriptstyle{\PICK,\ML}}-\pickainv\|_\infty\conP 0$ and
$\|\estpickin^{\scriptscriptstyle{\CFG,\ML}}-\pickainv\|_\infty\conP 0$ as $n\to\infty.$

Concluding, we study the case when $\alpha$ is estimated with the GPWM estimator in \eqref{eq:gpwm}
and $\picka$ with the P and CFG estimators in \eqref{eq:pick_est_pick} and \eqref{eq:pick_est_cfg}.
Assuming in addition to the previous case that $\alpha>1/(k-1)$, then by Proposition \ref{pro:asy_norm_gpwm} we have that in $\ell^{\infty}(\simp)$
$$
\sqrt{n}\left\{\estpickin^{\scriptscriptstyle{\circ,\GPWM}}(\cdot)-\pickainv(\cdot)\right\}\rightsquigarrow
(\phi_{\circ,\GPWM}(\cproc_Q))(\cdot),\quad n\to\infty,
$$
where $\phi_{\circ,\GPWM}(f)$, $f \in \ell^{\infty}([0,1]^{d+1})$, is given in Definition \ref{defi:phi_symbols}\ref{def:phi_GPWM}.
Ultimately, from this result and the functional version of Slutsky's lemma, it follows
that 
$$
\|\estpickin^{\scriptscriptstyle{\PICK,\GPWM}}-\pickainv\|_\infty\conP 0, \quad 
\|\estpickin^{\scriptscriptstyle{\CFG,\GPWM}}-\pickainv\|_\infty\conP 0, \quad n\to\infty.
$$

%
%

%
%

%
\subsubsection{Auxiliary results}\label{sub:auxiliary_results}
%

%
%

%
\begin{lem}\label{lem:mado_est_prob_rep}
Under Condition \ref{cond:cond_theo}\ref{en:first_cond} we have
$$
\sqrt{n}\left\{\ln \estpicka^\MD (\cdot)- \ln\picka(\cdot)\right\}=(\phi_{\MD}(\cproc_{G_\alpha,n}))(\cdot) +o_p(1),
$$
where $\phi_{\MD}$ is given in Definition \ref{defi:phi_symbols}\ref{def:phi_MD}.
\end{lem}
\begin{lem}\label{lem:covariance}
Let $(\bU, V)$ be defined as in \eqref{eq:unif_rvs}.
Let $\omega'_{\epsilon,\bu}$ be the function defined in \eqref{eq:weighted_omega2_function}
and $\varphi_{\ML}$ in \eqref{eq:phi_mle}. Then, for every $\epsilon\in[0,1/2)$ and
$\bu \in (0,1]^d \setminus \{\bone\}$, we have $ \EE{\omega'_{\epsilon,\bu}(\bU)\varphi_{\ML}(V)} \inr $, i.e. the expectation is finite.
\end{lem}
\begin{lem}\label{lem:mado_est_AS}
Let $\estalpha$ be an estimator of $\alpha$ satisfying $\estalpha \conAS \alpha$ as $n\to\infty$
%
%
%
%
and $\estpicka^{\MD}$ be the estimator of $\picka$ given in \eqref{eq:pick_est_md}. 
Let
$$
\estpickin(\bt)=
\left(\estpicka^{\MD}(\bt)/\norm{\bt}_{1/\estalpha}\right)^{1/\estalpha},\quad \bt\in\simp.
$$
Then, 
\begin{equation}\label{eq:est_pick_alpha_AS}
\left\|
\estpickin-\pickainv
\right\|_{\infty} \conAS 0,\quad n\to\infty.
\end{equation}
%

%
%
%
\end{lem}
\begin{lem}\label{lem:third_step}
Assume that $\alpha> 1/(k-1)$. Then, almost surely as $n\to \IF$
$$
\sqrt{n}(\estalpha^{\GPWM}-\alpha)=\tau_{\GPWM}(\cbbridge_{n})  + o(1),
$$
where $\tau_{\GPWM}:\ell^{\infty}([0,1])\to\R$ is defined as
\begin{eqnarray}\label{eq:varphi_GPWM}
\tau_{\GPWM}(f)&=&-2\int_0^1f(v)\frac{v(-\ln v)^k\{\mu_{1,k-1} - \mu_{1,k}/(-\ln v)\}}{\dot{\Phi}_\alpha(\Phi_\alpha^{\leftarrow}(v))(k\mu_{1,k-1}-2\mu_{1,k})^2}\diff v,\\
\nonumber \mu_{a,b}&=&\int_0^1 \Phi_\alpha^{\leftarrow}(v) v^{a}(-\ln v)^b \diff v, \quad a,b\in\N
\end{eqnarray}
%
%
%
%
and $\dot{\Phi}_\alpha(v)=\partial/\partial v \, \Phi_\alpha(v)$, $v\in(0,1)$.
\end{lem}
For the proofs see Section 4 of the supplementary material.

%
%

%
%

%

\section*{Acknowledgements}

S. A. Padoan is supported by the Bocconi Institute for Data Science and Analytics (BIDSA). 
E. Hashorva is supported by SNSF Grant 200021-175752/1.

\bibliographystyle{imsart-nameyear}

\end{document}